\documentclass[letterpaper]{article} 
\usepackage{aaai2026}  
\usepackage{times}  
\usepackage{helvet}  
\usepackage{courier}  
\usepackage[hyphens]{url}  
\usepackage{graphicx} 
\usepackage{natbib}  
\usepackage{caption} 
\usepackage{algorithm}
\usepackage{algorithmic}
\usepackage{amssymb}
\usepackage{amsthm}
\usepackage[table]{xcolor}
\usepackage{afterpage}
\usepackage{placeins}
\usepackage{newfloat}
\usepackage{listings}

\usepackage{amsmath}
\usepackage{booktabs}

\newtheorem{property}{Property}

\newtheorem{proposition}{Proposition}

\newcommand{\ie}{\textit{i.e.}\ }

\DeclareCaptionStyle{ruled}{labelfont=normalfont,labelsep=colon,strut=off} 
\floatstyle{ruled} 
\newfloat{listing}{tb}{lst}{}
\floatname{listing}{Listing}
\title{Improving Sparse IMU-based Motion Capture with Motion Label Smoothing}
\author{
    Zhaorui Meng\textsuperscript{\rm 1},
    Lu Yin\textsuperscript{\rm 1},
    Yangqing Hou\textsuperscript{\rm 1},
    Anjun Chen\textsuperscript{\rm 1}\thanks{Corresponding author.},
    Shihui Guo\textsuperscript{\rm 1},
    Yipeng Qin\textsuperscript{\rm 2}
}
\affiliations{
    \textsuperscript{\rm 1}Xiamen University,
   \textsuperscript{\rm 2}Cardiff University\\
   mengzhaorui@stu.xmu.edu.cn, yinlu@stu.xmu.edu.cn, 
   38120241150235@stu.xmu.edu.cn,
    anjunchen@xmu.edu.cn,
    guoshihui@xmu.edu.cn,
    qiny16@cardiff.ac.uk
}

\usepackage{bibentry}

\begin{document}

\maketitle

\begin{abstract}
Sparse Inertial Measurement Units (IMUs) based human motion capture has gained significant momentum, driven by the adaptation of fundamental AI tools such as recurrent neural networks (RNNs) and transformers that are tailored for temporal and spatial modeling. Despite these achievements, current research predominantly focuses on pipeline and architectural designs, with comparatively little attention given to regularization methods, highlighting a critical gap in developing a comprehensive AI toolkit for this task.
To bridge this gap, we propose {\it motion label smoothing}, a novel method that adapts the classic label smoothing strategy from classification to the sparse IMU-based motion capture task. 
Specifically, we first demonstrate that a naive adaptation of label smoothing, including simply blending a uniform vector or a ``uniform'' motion representation (e.g., dataset-average motion or a canonical T-pose), is suboptimal; and argue that a proper adaptation requires increasing the {\it entropy} of the smoothed labels.
Second, we conduct a thorough analysis of human motion labels, identifying three critical properties: 1) Temporal Smoothness, 2) Joint Correlation, and 3) Low-Frequency Dominance, and show that conventional approaches to entropy enhancement (e.g., blending Gaussian noise) are ineffective as they disrupt these properties.
Finally, we propose the blend of a novel skeleton-based Perlin noise for motion label smoothing, designed to raise label entropy while satisfying motion properties. Extensive experiments applying our motion label smoothing to three state-of-the-art methods across four real-world IMU datasets demonstrate its effectiveness and robust generalization (plug-and-play) capability.

\end{abstract}

\section{Introduction}
Human motion capture plays a critical role in diverse domains, including film production~\cite{menache2000understanding}, interactive gaming~\cite{geng2003reuse}, and medical rehabilitation~\cite{mousavi2014review}. Recently, sparse Inertial Measurement Units (IMUs) based motion capture systems have emerged as a lightweight yet promising alternative.

These systems achieve real-time human motion reconstruction using only six IMUs strategically positioned on the 

\begin{figure}[H] 
  \centering
  \includegraphics[width=0.9\linewidth]{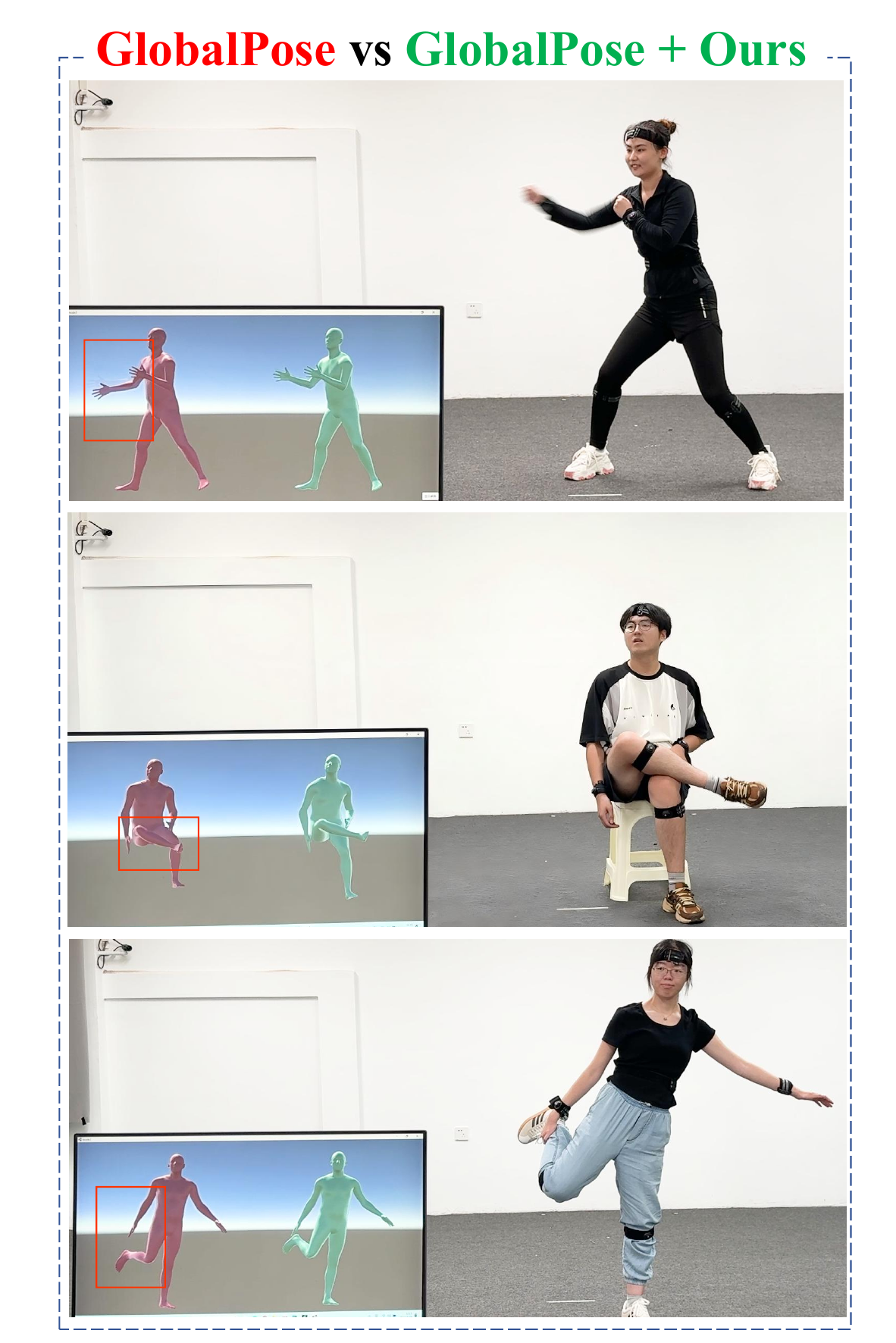}
  \caption{A live comparison between the state-of-the-art sparse IMU-based motion capture system, GlobalPose~\cite{yi2025improving} (left, {red}), and its improved variant enhanced with our motion label smoothing technique (right, {green}) clearly illustrates the effectiveness of our method.}

  \label{fig:teaser}
\end{figure}

\noindent wrists, ankles, head, and hips. This minimal configuration offers compelling advantages in portability, affordability, and resilience to occlusions or lighting variations, making them highly suitable for ubiquitous motion capture scenarios.

Fueled by recent advances in AI, sparse IMU-based motion capture has made remarkable progress. Early methods~\cite{huang2018deep, yi2021transpose} leveraged RNNs to reconstruct human motion; TIP~\cite{jiang2022transformer} introduced transformer architectures to improve accuracy; PIP~\cite{yi2022physical} enhanced RNNs with hidden-state initialization to disambiguate complex motions; PNP~\cite{yi2024pnp} calibrated acceleration signals via an autoregressive MLP. While effective, these approaches largely focus on adapting core neural architectures for temporal and spatial modeling to the task. In contrast, regularization, an equally critical component of deep learning, remains largely unexplored, revealing a key gap in building a more comprehensive AI toolkit for sparse IMU-based motion capture.

In this paper, we address this gap by introducing {\it motion label smoothing}, an adaptation of the classic label smoothing technique~\cite{szegedy2016rethinking} tailored for sparse IMU-based motion capture.
While it may appear straightforward, this adaptation poses significant challenges.
Specifically, a naive adaptation of label smoothing from classification tasks~\cite{szegedy2016rethinking, muller2019does}, such as blending the ground-truth label with a uniform label vector or a ``uniform'' motion representation (e.g., dataset-average motion or a canonical T-pose), proves suboptimal. We argue that this stems from a fundamental misinterpretation of ``smoothness'' in label smoothing: it is meant to increase label entropy, not merely enforce uniformity across label vectors. In classification, incorporating a uniform vector supports this objective; in motion capture, however, a ``uniform'' motion collapses into a static pose (e.g., a T-pose), paradoxically reducing entropy rather than enhancing it.
Therefore, to properly adapt label smoothing for sparse IMU-based motion capture, we first conduct a rigorous analysis of motion labels, identifying their three key properties: (1) Temporal Smoothness: motion evolves continuously over time; (2) Joint Correlation: rotations of adjacent joints within a kinematic chain are inherently linked; and (3) Low-Frequency Dominance: joint rotation signals in Euclidean space are dominated by low-frequency components.
Building on these properties, we demonstrate that naive entropy-enhancement strategies, such as blending Gaussian or uniform noise, are ineffective as they inevitably disrupt these intrinsic characteristics.
To address this challenge, we propose a novel skeleton-based Perlin noise method for motion label smoothing. 
Specifically, we first map joint rotations from the SO(3) manifold to a tractable Euclidean R6D representation~\cite{zhou2019continuity}, where these motion properties can be explicitly modeled. In this space, we construct a structured Perlin noise field whose spatial distribution encodes joint correlations (via kinematic chains) and whose temporal continuity preserves smoothness while still raising label entropy. Overlaying this skeleton-based noise onto motion labels yields smoothed labels that adhere to the principles of label smoothing while respecting the properties of human motion. We conduct comparison experiments across three state-of-the-art sparse IMU-based motion capture models and four real-world IMU datasets. 
Empirically, we also compare our method to naive adaptations of label smoothing and other label modification strategies.
Extensive experimental results demonstrate that our method consistently improves motion capture performance and outperforms competing strategies.

In summary, our contributions are:
\begin{itemize}
    \item We propose {\it motion label smoothing}, a novel adaptation of the classic label smoothing technique specifically tailored for sparse IMU-based motion capture, featuring the blending of a skeleton-based Perlin noise to motion data.
    \item To justify its necessity, we identify why a naive adaptation of label smoothing from classification is suboptimal, showing that they misinterpret ``smoothness'' as uniformity rather than entropy enhancement, leading to the misuse of static, low-entropy motions instead of meaningful regularization.
    \item In addition, we conduct the first rigorous analysis of motion labels for sparse IMU-based capture, identifying their three key properties: (1) Temporal Smoothness, (2) Joint Correlation, and (3) Low-frequency Dominance; and demonstrate that naive entropy-enhancement strategies (e.g., Gaussian or uniform noise) are also ineffective as they inevitably disrupt these properties.
    \item Extensive experiments on three state-of-the-art sparse IMU-based motion capture models and four real-world datasets demonstrate both the effectiveness and strong generalization capability of our method.
\end{itemize}
\section{Related Work}

\subsection{Sparse IMU-based Motion Capture}
Sparse IMU-based motion capture reconstructs human poses by estimating the local rotations of the 24 SMPL~\cite{loper2023smpl} joints through inverse kinematics, using signals from six IMUs placed on the left forearm, right forearm, left lower leg, right lower leg, head, and hips, respectively. 
However, this task remains highly challenging due to the sparsity of IMU data and its inherent noise. To address these challenges, prior works have primarily focused on adapting fundamental AI tools to enrich the toolkit for this task. Specifically, early solutions leveraged Recurrent Neural Networks (RNNs)~\cite{huang2018deep} to perform end-to-end pose estimation. Subsequent work by~\cite{yi2021transpose} broke the task into multiple prediction stages, enhancing pose accuracy and enabling global motion tracking with additional RNNs. Studies by~\cite{jiang2022transformer, wu2024accurate} explored the Transformer~\cite{vaswani2017attention} framework as an additional tool for inertial motion capture, while~\cite{yi2022physical} integrated physical optimization to improve the physical plausibility of predicted motions. Additionally,~\cite{yi2024pnp} mitigated the impact of non-inertial acceleration at the root joint during motion by adapting a self-supervised MLP network. Other efforts have targeted IMU drift and calibration errors, ~\cite{zuo2025transformer} introduced a real-time IMU calibration technique using a Transformer-based calibrator network, and~\cite{shao2025magshield} addressed magnetic interference in sparse IMU-based motion capture systems with an auxiliary LSTM to correct orientation errors leveraging human motion priors. Most recently, \cite{yi2025improving} achieved translation estimation in the full 3D space using physical optimization and refined pose estimation, making it state-of-the-art in sparse-IMU motion tracking. 

Nevertheless, while these methods have collectively advanced the AI toolkit for sparse IMU-based motion capture, they have predominantly focused on pipeline and architectural designs, with comparatively little attention given to {\it regularization} methods, revealing a critical gap in developing a comprehensive AI toolkit for this task.

\subsection{Label Smoothing}

Label smoothing \cite{szegedy2016rethinking} is a widely used regularization technique for improving the performance of deep learning models, with applications spanning diverse domains such as image classification, machine translation, and speech recognition \cite{chorowski2016towards, real2019regularized, huang2019gpipe, wei2021smooth, zhou2021rethinking, zhou2025training}. In practice, label smoothing operates in classification tasks by replacing one-hot labels with a softened target distribution formed by blending the ground truth label with a uniform label vector, which is commonly understood as a regularization technique designed to prevent models from becoming overly confident in their predictions and to improve their generalization capabilities \cite{pereyra2017regularizing, muller2019does}. Additionally, \cite{lukasik2020does} has highlighted its effectiveness in coping with label noise; \cite{yuan2020revisiting} further notes that knowledge distillation represents a form of learned label smoothing regularization, wherein label smoothing regularization serves as a virtual teacher model for knowledge distillation; \cite{zhang2021delving} proposes generating soft labels based on statistical predictions of the model for the target class, thereby implementing label smoothing. Furthermore, \cite{lienen2021label} introduces label relaxation, where the target is represented as a set of probabilities defined by an upper probability distribution. \cite{keriven2022not} explores smoothing in graph neural networks. 

Nevertheless, despite extensive research on label smoothing, its potential in sparse IMU-based motion capture remains unexplored.

\section{Preliminaries}

\subsubsection{Sparse IMU-based Motion Capture.} Our task involves predicting human motion based on real-time measurements from six IMUs attached to six body locations: 
\begin{equation}
    R = \texttt{Poser}(A_{S},O_{S},\omega_S^{root})
    \label{eq:sparse_imu_mocap}
\end{equation}
where \texttt{Poser} is the pose estimation network; acceleration ${A_{S}}\in \mathbb{R}^{3}$, orientation $ O_{S} \in \mathrm{SO}(3)$, and angular velocity of the root joint ${\omega_S^{root}}\in \mathbb{R}^{3} $ are the input measured by the six IMUs; the output $R$ comprises the rotations of the 24 joints of the SMPL \cite{loper2023smpl} human model (body pose). 

\subsubsection{Label Smoothing.} 
As defined in~\cite{szegedy2016rethinking}, given a ground-truth label \( y \), label smoothing modifies \( y \) into \( y' \), which comprises a mixture of a vector \( u \) and \( y \) weighted by \( 1 - \epsilon \) and \( \epsilon \), respectively:
\begin{equation}
\label{eq:label_smoothing}
    y' = (1 - \epsilon) y + \epsilon u
\end{equation}
where \( u \) is a uniform vector of value \( \frac{1}{K} \), where \( K \) represents the number of classes in a classification task.

\subsubsection{Motion Label Smoothing.} We adapt label smoothing (Eq.~\ref{eq:label_smoothing}) to our task (Eq.~\ref{eq:sparse_imu_mocap}) by replacing the $y$ in Eq.~\ref{eq:label_smoothing} with the ground truth rotation $R$ in Eq.~\ref{eq:sparse_imu_mocap}:
\begin{equation}
\label{eq:motion_label_smoothing}
    R' = (1 - \epsilon) R + \epsilon u
\end{equation}
Nevertheless, this adaptation is challenging due to the choice of $u$, which will be discussed in detail below.

\section{Analysis}

\begin{figure*}[t]
	\centering
 \includegraphics[width=1\linewidth]{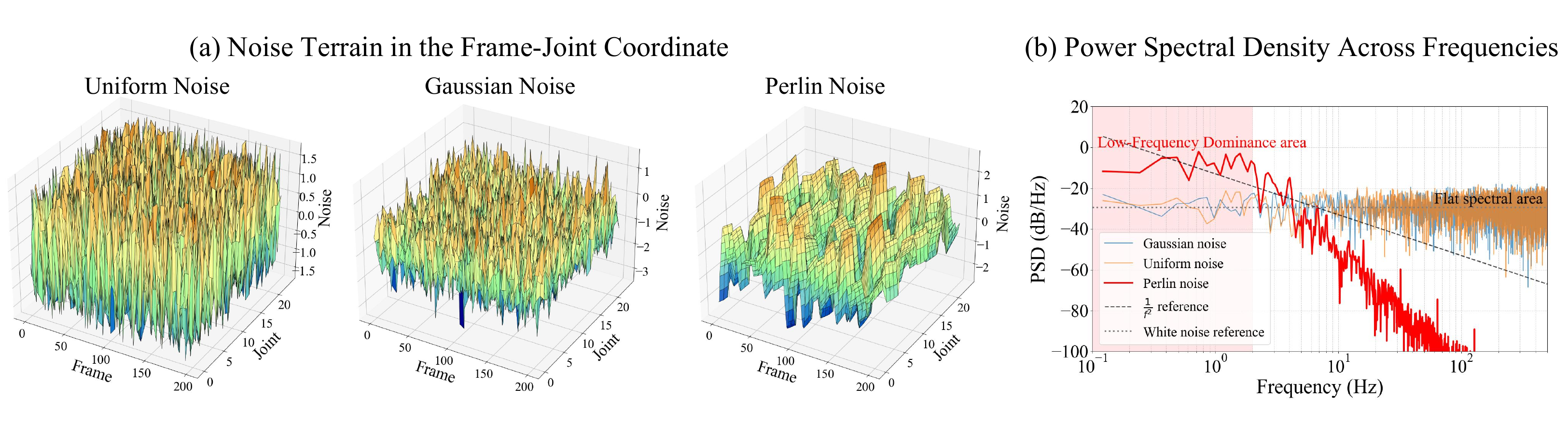}
	\caption{(a) Noise terrain of three types of noise in the frame-joint coordinate system, where only Perlin noise exhibits continuity across frames and correlation across joints. (b) Power spectral density (PSD) of the three types of noise, where only Perlin noise is dominated by low-frequency components.
}
	\label{fig:noise}
\end{figure*}

\subsection{Naive Adaptation of Label Smoothing is Ineffective}

Despite its effectiveness, the interpretation of label smoothing remains ambiguous, and we argue that the naive adaptation of it would undermine its potential as such an adaptation stems from a misinterpretation.

\subsubsection{Interpretation Ambiguity of Label Smoothing.} Despite its clear definition (Eq.~\ref{eq:label_smoothing}), label smoothing has two seemingly equally valid interpretations regarding the choice of $u$:
\begin{itemize}
    \item \textit{Uniform $u$ (naive interpretation).} As indicated by the use of the symbol ``$u$'', label smoothing works by increasing the {\it uniformity} of $y'$ using a high-uniformity $u$ which distributes label probability uniformly across all classes.
    \item \textit{High-entropy $u$ (our interpretation).} Unlike the naive interpretation, we argue that label smoothing works by increasing the entropy of $y'$ using a high-entropy $u$.
\end{itemize}
Thus far, it is difficult to assess the validity of these two interpretations, as they result in the same $u$ for classification tasks, i.e., a uniform vector of value \( \frac{1}{K} \) as mentioned above.

\subsubsection{Naive Adaptation is Ineffective.} 
\label{sec:naive}
Nonetheless, these two interpretations would lead to {\it different} $u$ when adapting label smoothing to sparse IMU‑based motion capture: 
\begin{itemize}
    \item \textit{Naive adaptation.} To ensure uniformity, \( u \) should represent an ``average'' motion label, e.g., the T-Pose or the mean motion derived from the training dataset.
    \item \textit{Our adaptation.} To increase entropy, \( u \) should be a noise vector tailored for motion representations.
\end{itemize}

Empirically, through extensive experiments (Table~\ref{tab:alter}) and accompanying discussions, we demonstrate that (1) the naive adaptation is ineffective, revealing that the naive interpretation is indeed a misinterpretation; (2) our interpretation is valid, offering new insights into the underlying mechanisms of label smoothing.

\subsection{Naive Entropy-Enhancement Strategies are Ineffective}

As mentioned above, our adaptation is non‑trivial as it requires the design of noise vectors tailored to motion representations. To address this, we first conduct a rigorous analysis of motion labels to identify three key properties, and then demonstrate that naive noise designs (e.g., Gaussian) are ineffective as they disrupt these properties.

\begin{property}[Temporal Smoothness]
    Human motion is constrained by muscle strength, skeletal structure, joint range of motion, and inertia, which inherently precludes abrupt changes in joint rotations $R$ over short time intervals, i.e., $\|\omega(t)\|$ is bounded:
    \begin{equation}
    \omega(t) \approx \frac{R(t + \Delta t) - R(t)}{\Delta t}, \quad \|\omega(t)\| \leq M,
    \end{equation}
    where $t$ is time; $\omega(t)$ is the first derivative of $R$ over $t$; \( \|\cdot\| \) is a vector norm (e.g., L2-norm). \( M \) represents a physiological upper bound derived from biomechanical constraints.
    \label{property1}
\end{property}

\begin{property}[Joint Correlation]
The human body is a highly coupled system of rigid skeletal chains, where proximal joint motion constrains distal joints (e.g., elbow rotation influences the shoulder). This structure can naturally be modeled as a skeleton tree, with the waist node as the root and edges representing joint dependencies and biomechanical constraints.
Then, for any parent-child joint pair \( (j_{\text{parent}}, j_{\text{child}}) \), the rotation range of $j_{\text{child}}$ is constrained by the rotation angle of $j_{\text{parent}}$ (e.g., the knee range of motion narrows when the hip is flexed):
\begin{equation}
{R}_{\text{child}}(t) \in A({R}_{\text{parent}}(t))    
\end{equation}
where the admissible set $A$ is defined as:
\begin{equation}
A = \{{R}_{\text{child}} | \phi_{\min}({R}_{\text{parent}}) \leq {R}_{\text{child}} \leq \phi_{\max}({R}_{\text{parent}})\}
\end{equation}
where $\phi$ is a Lipschitz continuous joint correlation function; \( \phi_{\min}, \phi_{\max} \) are its minimum and maximum values.
\label{property2}
\end{property}

\begin{property}[Low-Frequency Dominance]

Human motion data exhibits a predominance of low-frequency signals corresponding to normal movements, with high-frequency components (e.g., intense jitter) being relatively rare. 
Let $R(t) \in \mathbb{R}^d$ denote the $d$-dimensional motion label at time $t$. 
We define low-frequency dominance as: $\exists$ a frequency threshold $f_c \ll f_{\max}$ for all channels $c \in \{1, 2, ..., d\}$ and a ratio threshold $\alpha \gg 0$ such that:
\begin{equation}
\frac
{\sum_{c=1}^{d} \int_{0}^{f_c} P_c(f)  df}
{\sum_{c=1}^{d} \int_{0}^{f_{\max}} P_c(f)  df} 
\geq \alpha,
\end{equation}
where $f_{\max}$ is the Nyquist frequency ($f_{\max} = \frac{1}{2\Delta t}$ and $\Delta t$ is the IMU sampling interval); the power spectral density (PSD) $P_c(f)$ for each channel $c$ is computed as:
\begin{equation}
    P_c(f) = \left| \int_{0}^{T} R_c(t) e^{-i2\pi f t}  dt \right|^2
\end{equation}
where $T$ denotes the number of frames sampled by IMU.
In our experiments, we have $f_c=5\text{Hz}$ when $\alpha=0.7$.
\label{property3}
\end{property}

\subsubsection{Invalidity of Naive Strategies.} 

Given the above properties, we show that naive entropy-enhancing strategies, such as implementing $u$ (Eq.~\ref{eq:motion_label_smoothing}) as a Gaussian or uniform noise, are ineffective as they disrupt these properties. Specifically, since Gaussian and uniform noise are {\it independent and identically distributed (i.i.d.)}, they have:
\begin{itemize}
    \item {\bf An inherent trade-off between noise amplitude, and temporal smoothness together with joint correlation} (Properties~\ref{property1},~\ref{property2}). Sine an \textit{i.i.d.} distribution lacks dependencies across both temporal and joint dimensions, increasing the noise amplitude (necessary for effective regularization) inevitably amplifies discrepancies along these dimensions and violates Properties~\ref{property1} and~\ref{property2}.
    \item {\bf Flat power spectral densities (PSDs)} containing significant high-frequency components, contradicting Prop.~\ref{property3}.
\end{itemize}

\section{Method}
\label{sec:Method}
As mentioned above, adapting label smoothing to our task (Eq.~\ref{eq:motion_label_smoothing}) is challenging, as its $u$ must satisfy Properties~\ref{property1},~\ref{property2},~\ref{property3}, as well as having a sufficiently large noise amplitude for effective regularization.
 
To address this challenge, we propose a novel design that implements $u$ as a {\it skeleton‑based Perlin noise} as follows, ensuring continuity and smoothness while eliminating sharp discontinuities.

\subsection{Skeleton-based Perlin Noise}
\label{sec:Ske-Perlin}
Specifically, we define our skeleton‑based Perlin noise as:
\begin{equation}
    u = \texttt{sk-Perlin}(JC, \mathcal{H}, {size})
\end{equation}
where $JC$ represents the six joint chains defined from the SMPL~\cite{loper2023smpl} skeleton (\ie, left leg, right leg, left arm, right arm, torso, and head), with the six IMUs attached to their terminal joints; $\mathcal{H}=\{S_b,S_t,S_s,p,oct,l\}$ denotes the basic parameters in constructing the Perlin nosie~\cite{perlin1985image}, including base scale $S_b$, time scale $S_t$, space scale (joint scale) $S_s$, persistence $p$, octaves $oct$ and lacunarity $l$; and $size$ denotes the dimensional extent of each spatio-temporal axis, which we set to same dimension as the ground truth label.

\begin{figure}[t]
	\centering
 \includegraphics[width=0.95\linewidth]{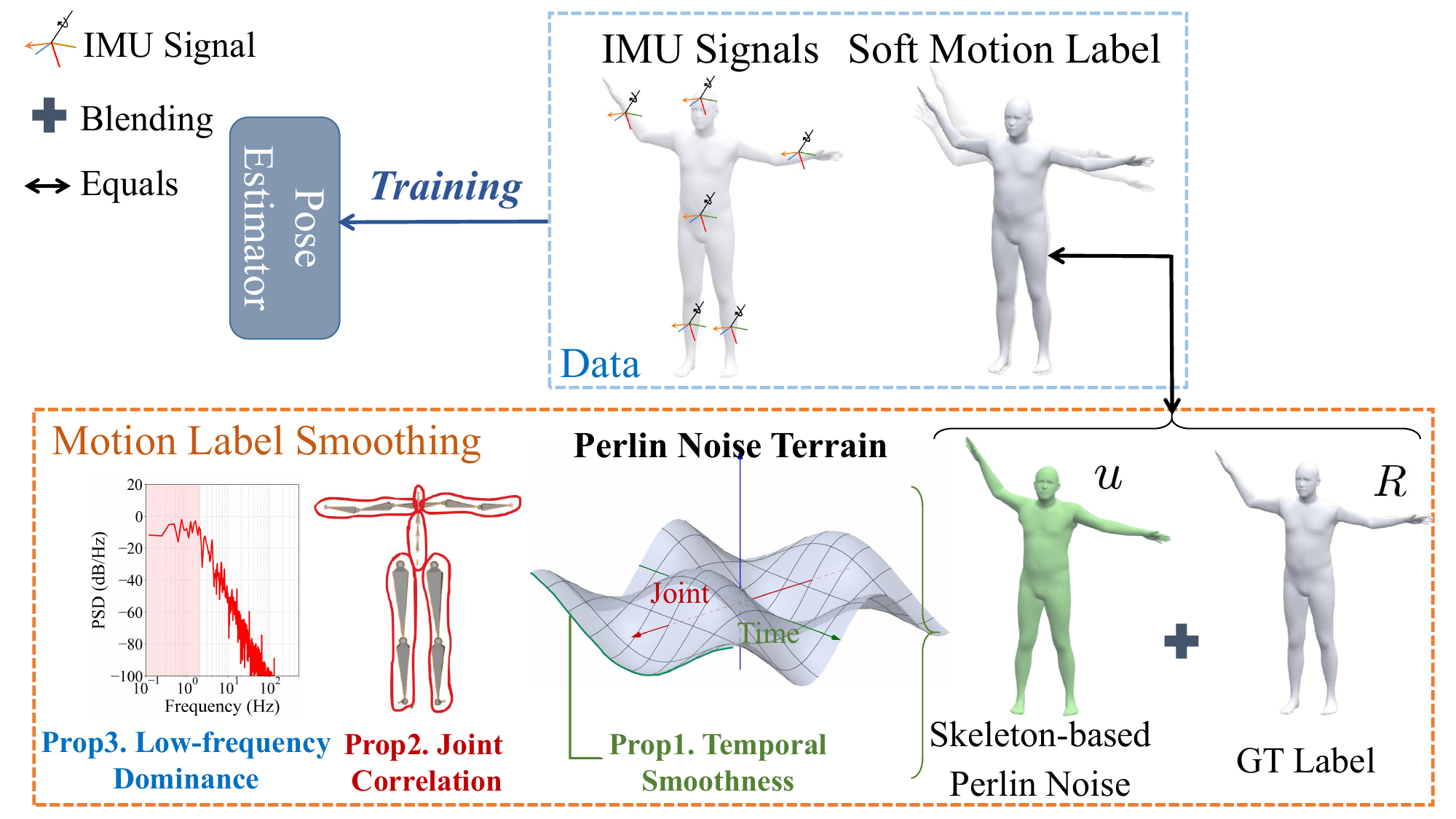}
 \caption{Overview of our \textit{Motion Label Smoothing} method. The ground truth motion label $R$ is blended with a carefully-constructed skeleton-based Perlin noise $u$, which satisfy Properties~\ref{property1},~\ref{property2},~\ref{property3}, as well as having a sufficiently large noise amplitude for effective regularization. }
\label{fig:Pipeline}
\end{figure}

\subsection{Amplitude-Decoupling and Properties-Satisfying Design}

Perlin noise~\cite{perlin1985image} creates smooth, natural textures and is widely utilized in computer graphics and natural phenomenon simulation. Its construction process involves dividing the three-dimensional space (\(time\), \(joint\), and \(channel\) in our \(\texttt{sk-Perlin}\)) into a grid of equal intervals, representing coordinates along the three dimensions of a given size, segmented by points \((x, y, z)\). At each grid point, a random unit gradient vector \(\mathbf{g}_{x,y,z}\) (with length 1 and random direction) is assigned. The influence at point \((x, y, z)\) is computed as the dot product of the gradient vectors from the surrounding eight grid points. Subsequently, an interpolation function smooths the values across these grid points. Enhanced detail is achieved by superimposing multiple noise layers of varying frequencies and amplitudes, known as octaves $oct$.

\subsubsection{Amplitude-Decoupling.}
The amplitude of the base noise, controlled by \( S_b \), directly determines the overall magnitude of our \(\texttt{sk-Perlin}\), while the interpolation function governs the smoothness of transitions between adjacent grid points. Thus, the amplitude and smoothness of Perlin noise are decoupled and controlled by distinct parameters.
In other words, the interpolated nature of Perlin noise (unlike the {\it i.i.d.} Gaussian or uniform noises) ensures smoothness while maintaining an effective noise amplitude \( S_b \).
 
\paragraph{Satisfaction of Properties 1 and 2.} As outlined above, interpolation along the \( x \)-direction ensures smoothness and continuity in the temporal dimension, with the temporal scaling factor \( S_t \) (temporal frequency) stretching the \( x \)-axis to enhance smoothness. Additionally, as illustrated in Fig.~\ref{fig:Pipeline}, our skeleton-based Perlin approach initially applies a base noise to each joint chain, followed by a single-octave noise to each joint within the chain, stripping all high-frequency details to produce ultra-smooth offsets. The final Perlin noise is synthesized by combining and scaling the base noise with the offsets. This method guarantees that the noise within each joint chain exhibits correlation (derived from the base noise) while distinguishing individual joints, thereby satisfying Property~\ref{property2}. We plot the different noises in the \( time-joint\) coordinate in Fig.~\ref{fig:noise} (a), clearly showing that our method, based on interpolation, ensures continuity in both temporal and spatial (joint) dimensions.

\paragraph{Satisfaction of Property 3.} The low-frequency dominance of Perlin noise arises from the low-resolution gradient field of the base grid and the attenuation of high-frequency components. The amplitude of octave superposition decays exponentially with frequency (\( \frac{1}{2^i} \)), while the persistence \( p \) regulates the contribution of high frequencies. We set a low octave count \( oct = 5 \) to reduce the number of high-frequency layers and adjust \( p = 0.5 \) to attenuate high-frequency effects, ensuring that low-frequency components predominate. This effect is validated through power spectral density (PSD) analysis, as shown in Figure~\ref{fig:noise} (b). Visually, while Gaussian and uniform noise display similar spectral densities across all frequencies, Perlin noise exhibits significantly higher density in low-frequency bands, with a monotonic decrease as frequency increases.

\begin{table*}[ht]
    \centering
    
    \fontsize{9.50}{9}\selectfont
    \begin{tabular*}{\textwidth}{lcccccc}
        \toprule
         & \multicolumn{3}{c}{\textbf{TotalCapture}} & \multicolumn{3}{c}{\textbf{ANDY}} \\
        \cmidrule(lr){2-4} \cmidrule(lr){5-7}
        \textbf{Method} & \textbf{SIP Err} & \textbf{Joint Err} & \textbf{Mesh Err} & \textbf{SIP Err} &  \textbf{Joint Err} & \textbf{Mesh Err} \\
        \midrule
        \textit{TransPose} & 14.28 &  5.31 & 5.89 & 31.91 &  13.77 & 18.96 \\
        \textit{TP+Ours} & 12.49 (\(\downarrow 12.54\%\)) & 5.00 (\(\downarrow 5.84\%\)) & 5.55 (\(\downarrow 5.77\%\)) & 31.20 (\(\downarrow 2.23\%\)) &  13.42 (\(\downarrow 2.54\%\)) & 18.46 (\(\downarrow 2.64\%\)) \\
        \midrule
        \textit{PIP} & 11.16 &  4.55 & 5.26 & 29.59 &  13.56 & 18.79 \\
        \textit{PIP+Ours} & 10.54 (\(\downarrow 5.56\%\))  & 4.38 (\(\downarrow 3.74\%\)) & 5.07 (\(\downarrow 3.61\%\)) & 29.04 (\(\downarrow 1.86\%\)) &  13.49 (\(\downarrow 0.52\%\)) & 18.67 (\(\downarrow 0.64\%\)) \\
        \midrule
        \textit{GlobalPose} & 9.85 &  3.96 & 4.35 & 39.58 &  17.66 & 23.45 \\
        \textit{GP+Ours} & 7.84 (\(\downarrow 20.41\%\))  & 3.26 (\(\downarrow 17.68\%\)) & 3.75 (\(\downarrow 13.79\%\)) & 36.71 (\(\downarrow 7.25\%\)) &  16.92 (\(\downarrow 4.19\%\)) & 22.22 (\(\downarrow 5.25\%\)) \\
        \midrule
        & \multicolumn{3}{c}{\textbf{CIP}} & \multicolumn{3}{c}{\textbf{DIP-IMU}} \\
        \cmidrule(lr){2-4} \cmidrule(lr){5-7}
        \textbf{Method} & \textbf{SIP Err} & \textbf{Joint Err} & \textbf{Mesh Err} & \textbf{SIP Err} &  \textbf{Joint Err} & \textbf{Mesh Err} \\
        \midrule
        \textit{TransPose} & 28.46 &  10.82 & 11.91 & 14.04 &  4.86 & 5.80 \\
        \textit{TP+Ours} & 26.19 (\(\downarrow 7.97\%\))  & 10.65 (\(\downarrow 1.57\%\)) & 11.71 (\(\downarrow 1.68\%\)) & 13.57 (\(\downarrow 3.35\%\)) &  4.64 (\(\downarrow 4.53\%\)) & 5.50 (\(\downarrow 5.17\%\)) \\
        \midrule
        \textit{PIP} & 25.57 &  9.02 & 10.76 & 12.08 &  4.33 & 5.06 \\
        \textit{PIP+Ours} & 23.87 (\(\downarrow 6.65\%\)) &  8.60 (\(\downarrow 4.66\%\)) & 10.27 (\(\downarrow 4.55\%\)) & 11.62 (\(\downarrow 3.81\%\)) &  4.18 (\(\downarrow 3.46\%\)) & 4.88 (\(\downarrow 3.56\%\)) \\
        \midrule
        \textit{GlobalPose} & 23.04 &  6.98 & 8.13 & 13.77 &  4.36 & 5.08 \\
        \textit{GP+Ours} & 22.32 (\(\downarrow 3.12\%\)) &  6.47 (\(\downarrow 7.31\%\)) & 7.70 (\(\downarrow 5.29\%\)) & 13.50 (\(\downarrow 1.96\%\)) &  4.27 (\(\downarrow 2.06\%\)) & 4.98 (\(\downarrow 1.97\%\)) \\
        \bottomrule
    \end{tabular*}
    \normalsize
    \caption{Quantitative comparisons between baseline methods and those augmented with our motion label smoothing method with error percentage reduction (\(\downarrow \text{percentage}\%\)). 
    }
    \label{tab:comparison}
\end{table*}

\section{Experiments}

\subsection{Implementation Details}

\subsubsection{Training Setup and Datasets.}
Our method is designed as a plug-and-play training tool that requires no modification to model architectures. Existing approaches that innovate through model architecture have established a standardized training pipeline, utilizing the AMASS~\cite{mahmood2019amass} dataset and synthesized IMU data as the training set, followed by fine-tuning with real IMU datasets \cite{huang2018deep, trumble2017total}. Our experiments adhere to this same setup. Since our method modifies only the joint rotation data, fine-tuning is performed exclusively on the pose estimation networks, while parameters of other networks, such as those predicting joint velocities or foot-ground contact probabilities, are retained using publicly available pre-trained weights. All training parameters and details strictly follow the original implementations.

Our test set is selected from four real IMU datasets, including TotalCapture \cite{trumble2017total}, ANDY \cite{maurice2019human}, CIP \cite{palermo2022raw} and DIP-IMU \cite{huang2018deep}.

\subsubsection{Baseline Methods}

We evaluate the effectiveness of our method on three representative baseline pose estimation algorithms:
\begin{itemize}
	\item TransPose~\cite{yi2021transpose}, the first real-time algorithm for global human motion tracking using only six IMUs;
	\item PIP~\cite{yi2022physical}, the first method to incorporate physical constraints through optimization, whose framework has been adopted by many follow-up works; and
	\item GlobalPose~\cite{yi2025improving}, the most recent and state-of-the-art algorithm in this domain.
\end{itemize}

\subsubsection{Evaluation Metrics.}
We employ four standard evaluation metrics  used in exsiting works to assess the effectiveness of the methods. \textit{SIP Error (\(^{\circ}\))}~\cite{von2017sparse} measures mean global rotation error of hips and shoulders; \textit{Angular Error (\(^{\circ}\))} measures mean global rotation error of all joints; \textit{Positional Error (cm)} measures mean position error of all joints; \textit{Mesh Error (cm)} measures mean vertex error of the posed SMPL meshes. Lower values indicate higher motion capture accuracy.


 \begin{figure}[t]

	\centering
 \includegraphics[width=1\linewidth]{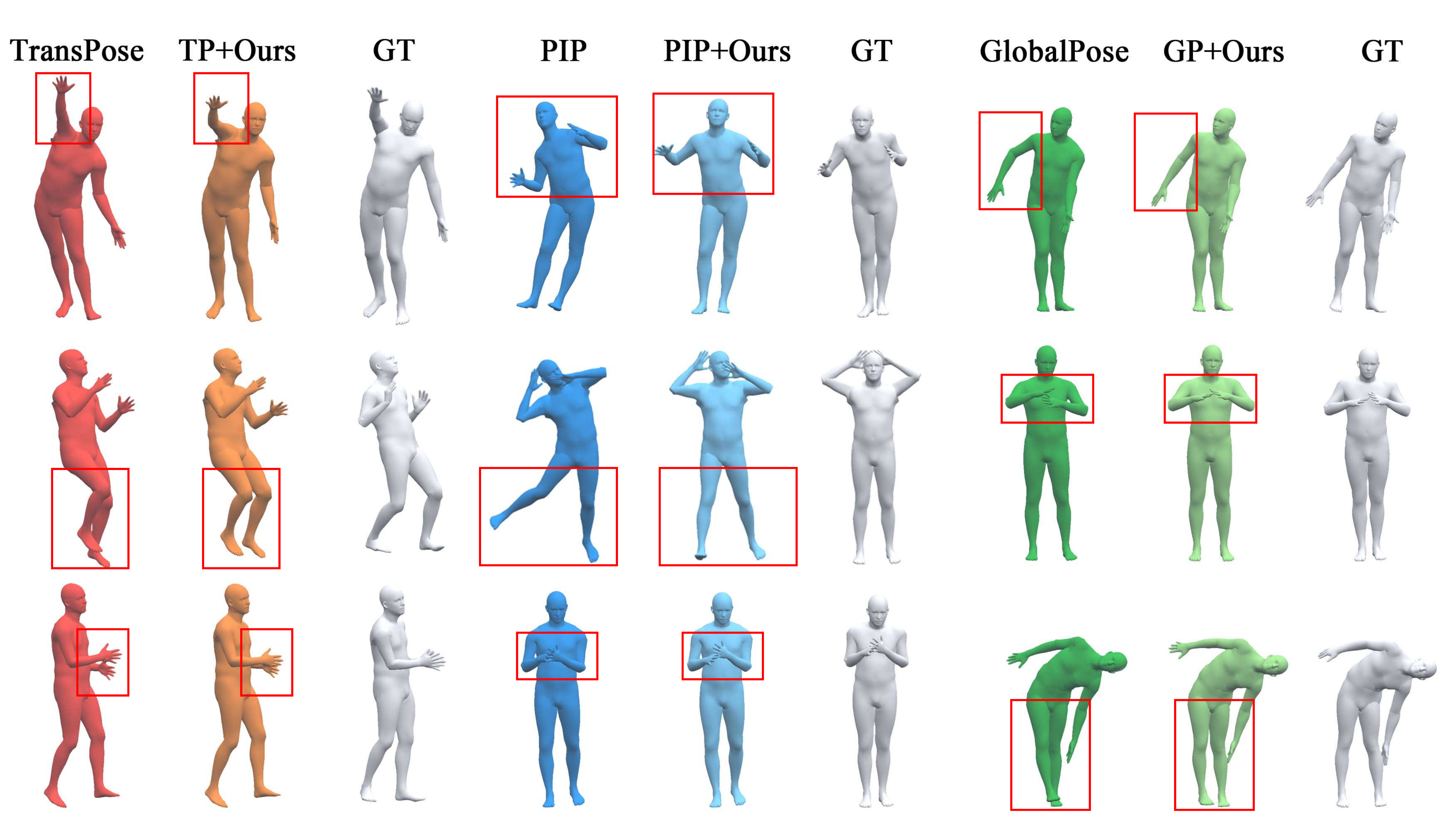}
	\caption{Qualitative comparisons with baseline methods. Examples are from the TotalCapture and CIP datasets. 
}
	\label{fig:Compare}
\end{figure}

\subsection{Comparisons}
\subsubsection{Quantitative Comparisons.} We compare the performance of baseline methods using conventional supervision with those fine-tuned using our approach across four test datasets. Table.~\ref{tab:comparison} reports the percentage reduction in error metrics for each of the three baseline algorithms after applying our method, showing consistent improvements across all settings. We attribute this to our method being the first regularization strategy specifically designed for sparse IMU-based motion capture. Its plug-and-play nature enables integration into a wide range of pose estimation pipelines, thereby expanding the current sparse-IMU based motion capture toolkit. Notably, we have advanced the state-of-the-art accuracy of GlobalPose, achieving a significant \( 20.41\% \) reduction in SIP Error on the TotalCapture dataset.

\subsubsection{Qualitative Comparisons.} Furthermore, we provide qualitative comparisons on the test data in Fig. \ref{fig:Compare}, revealing noticeable enhancements in actual motion quality when integrating our training method.  For instance, in the third example of TransPose (bottom-left panel of Fig.~\ref{fig:Compare}), our method yield improvements in the arm movements, transitioning hand positions from being misaligned across different levels to being aligned at the same level, bringing them closer to the ground truth compared to the original results. In the first example of GlobalPose (Fig.~\ref{fig:Compare}, top-right), the baseline method produces a bendy right arm that misaligns with the actual pose, whereas our method yields a straight arm that aligns accurately.

\subsection{Evaluation}

\subsubsection{Ablation Studies.}

We conduct ablation studies on the core properties of human motion addressed by our proposed method, following the logic outlined below. Using the baseline method as a control, we initially apply Gaussian noise as a naive entropy-enhancement label smoothing approach (\textit{Baseline w/ Label Smoothing}). Subsequently, guided by the properties we have defined, we refine the noise to align with human motion properties by sequentially incorporating: 1) Temporal Smoothness: We apply Gaussian smoothing to the noise in \textit{Baseline w/ Label Smoothing} (\textit{Baseline + T}); 2) Joint Correlation: We leverage the same joint chain constraints as detailed in the Method section to further enhance the noise's alignment with human motion properties (\textit{Baseline + T + J}); 3) Low-Frequency Dominance: We replace the Gaussian noise from step 2 with a Perlin noise field, designed with specific scale and persistence parameters to reflect low-frequency dominance (\textit{Baseline + T + J + L (ours)}). As shown in Table.~\ref{tab:ablation}, the naive label smoothing method provides a moderate improvement in the baseline model's performance. However, the sequential integration of the three defined properties further reduces errors, underscoring the significance of our proposed skeleton-based Perlin noise tailored to these properties.

\subsubsection{Alternative Design.}

In Table.~\ref{tab:alter}, we substantiate the necessity of our proposed motion label smoothing method by comparing it against alternative potential solutions, including the configurations below:

\begin{enumerate}
    \item \textbf{Naive Adaptation}: To evaluate the rationality of our approach of increasing label entropy, following the method outlined in the Naive Adaptation section, we replace the uniform vector \( u \) with two stationary motion vectors: the T-Pose and the average motion derived from the AMASS dataset \cite{mahmood2019amass}. This method maintains formal consistency with Eq.~\ref{eq:label_smoothing}, and we adjust \( \epsilon \) to 0.1 to align the magnitude of the added term with our noise blending scheme.

    \item \textbf{Entropy-Enhancement}: We compare our method with other approaches to increasing label entropy, including the addition of uniform noise and Gaussian noise with consistent intensity.

    \item \textbf{Alternatives}: Other potential label smoothing methods: 
    \begin{enumerate}
        \item \textit{Temporal Smoothing}: A naive temporal smoothing approach by directly applying Gaussian smoothing to the original labels.
        \item \textit{Knowledge Distillation}: Widely employed to enhance model performance, the soft labels produced by the teacher model in knowledge distillation act as an implicit form of label smoothing regularization (LSR)~\cite{yuan2020revisiting}. We utilize poses optimized with a physics-based module (proposed in \cite{yi2022physical} and refined in \cite{yi2025improving}), which incorporates physical information, as the applied distribution in Eq.~\ref{eq:label_smoothing}. This distribution is overlaid onto the ground truth to form a distillation-based LSR.
    \end{enumerate}
\end{enumerate}

Our method outperforms all the aforementioned solutions, a result we attribute to its motion-property-awareness. We analyze or mathematically demonstrate why these methods fall short of our approach, drawing on the human motion properties we have identified.

\begin{table}[ht]
    \centering
    \begin{tabular}{lcccc}
        \toprule
        Method & SIP& Ang & Joint & Mesh\\
        \toprule
        Baseline (GlobalPose) & 9.85 & 9.55 & 3.96 & 4.35\\
        \midrule
        \multicolumn{5}{l}{w/ Label Smoothing} \\
        \midrule
        Baseline & 8.82 & 8.65 & 3.82 & 4.43\\ 
        Baseline+\textit{T}  & 8.59 & 8.30 & 3.77 & 4.37\\
        Baseline+\textit{T}+\textit{J}  & 8.22 & 8.02 & 3.52 & 4.12\\
        Baseline+\textit{T}+\textit{J}+\textit{L}(ours)  & \textbf{7.84} & \textbf{7.87} & \textbf{3.26} & \textbf{3.75}\\
        \toprule
    \end{tabular}
    \caption{
    We examine the effectiveness of incorporating the proposed motion properties (Temporal Smoothness (\textit{T}), Joint Correlation (\textit{J}), and Low-frequency Dominance (\textit{L})).}
    \label{tab:ablation}
\end{table}

\begin{table}[ht]
    \centering
    \begin{tabular}{lcccc}
        \toprule
        Method & SIP& Ang & Joint & Mesh\\
        \toprule
        \multicolumn{5}{l}{Naive Adaptation} \\
        \midrule
        T-pose Vector & 8.97 & 8.97 & 3.96 & 4.73\\ 
        AMASS Mean Vector & 8.75 & 9.17 & 3.87 & 4.60\\
        \midrule
        \multicolumn{5}{l}{Entropy-Enhancement} \\
        \midrule
        Uniform Noise & 8.72 & 8.61 & 3.82 & 4.44\\ 
        Gaussian Noise & 8.82 & 8.65 & 3.82 & 4.43\\
        \midrule
        \multicolumn{5}{l}{Alternatives} \\
        \midrule
        Temporal Smoothing  & 8.23 & 8.07 & 3.57 & 4.15\\
        Distillation & 8.46 & 8.25 & 3.59 & 4.17\\ 
        \midrule
        Ours  & \textbf{7.84} & \textbf{7.87} & \textbf{3.26} & \textbf{3.75}\\
        \toprule
    \end{tabular}
    \caption{Quantitative comparison with alternative strategies.}
    \label{tab:alter}
\end{table}

\subsection{Limitations and Discussions}

While our method enhances existing sparse-IMU systems, it also inherits the inherent limitations of these methods. For instance, prior work relies on the template SMPL body model, overlooking the impact of body shape on IMU data, which hinders generalization to individuals with diverse body types, such as children or exceptionally tall subjects. However, given the plug-and-play nature of our motion label smoothing method, we anticipate that integrating shape-aware algorithms will effectively address this issue.
Moreover, existing methods leverage public datasets such as AMASS for training; although extensive, these datasets are still limited in the range of motion types, thereby posing challenges in reconstructing complex motions like slipping and street dance.

\section{Conclusion}

In this work, we introduce the first regularization tool for the sparse-IMU based motion capture AI toolkit. We initially show that a naive adaptation of traditional label smoothing is insufficient. Subsequently, we conduct a systematic study, identifying three inherent properties of human motion data that constrain the effectiveness of naive entropy-enhancement label smoothing techniques. Finally, we propose a novel task-specific motion label smoothing method, achieved by blending skeleton-based Perlin noise with ground-truth labels, and show its alignment with the identified properties. Extensive experiments confirm the effectiveness of our motion label smoothing method.

\section{Acknowledgments}
This work is supported by National Natural Science Foundation of China (62472364, 62072383), the Public Technology Service Platform Project of Xiamen City (No.3502Z20231043), Xiaomi Young Talents Program / Xiaomi Foundation and the Fundamental Research Funds for the Central Universities (20720240058), “Young Eagle Plan" Top Talents of Fujian Province. Anjun Chen is the corresponding author. 

\bibliography{ref}

\begin{thebibliography}{36}
\providecommand{\natexlab}[1]{#1}

\bibitem[{Chorowski and Jaitly(2016)}]{chorowski2016towards}
Chorowski, J.; and Jaitly, N. 2016.
\newblock Towards better decoding and language model integration in sequence to sequence models.
\newblock \emph{arXiv preprint arXiv:1612.02695}.

\bibitem[{Geng and Yu(2003)}]{geng2003reuse}
Geng, W.; and Yu, G. 2003.
\newblock Reuse of motion capture data in animation: A review.
\newblock In \emph{International Conference on Computational Science and Its Applications}, 620--629. Springer.

\bibitem[{Huang et~al.(2019)Huang, Cheng, Bapna, Firat, Chen, Chen, Lee, Ngiam, Le, Wu et~al.}]{huang2019gpipe}
Huang, Y.; Cheng, Y.; Bapna, A.; Firat, O.; Chen, D.; Chen, M.; Lee, H.; Ngiam, J.; Le, Q.~V.; Wu, Y.; et~al. 2019.
\newblock Gpipe: Efficient training of giant neural networks using pipeline parallelism.
\newblock \emph{Advances in neural information processing systems}, 32.

\bibitem[{Huang et~al.(2018)Huang, Kaufmann, Aksan, Black, Hilliges, and Pons-Moll}]{huang2018deep}
Huang, Y.; Kaufmann, M.; Aksan, E.; Black, M.~J.; Hilliges, O.; and Pons-Moll, G. 2018.
\newblock Deep inertial poser: Learning to reconstruct human pose from sparse inertial measurements in real time.
\newblock \emph{ACM Transactions on Graphics (TOG)}, 37(6): 1--15.

\bibitem[{Jiang et~al.(2022)Jiang, Ye, Gopinath, Won, Winkler, and Liu}]{jiang2022transformer}
Jiang, Y.; Ye, Y.; Gopinath, D.; Won, J.; Winkler, A.~W.; and Liu, C.~K. 2022.
\newblock Transformer inertial poser: Real-time human motion reconstruction from sparse imus with simultaneous terrain generation.
\newblock In \emph{SIGGRAPH Asia 2022 Conference Papers}, 1--9.

\bibitem[{Keriven(2022)}]{keriven2022not}
Keriven, N. 2022.
\newblock Not too little, not too much: a theoretical analysis of graph (over) smoothing.
\newblock \emph{Advances in Neural Information Processing Systems}, 35: 2268--2281.

\bibitem[{Lienen and H{\"u}llermeier(2021)}]{lienen2021label}
Lienen, J.; and H{\"u}llermeier, E. 2021.
\newblock From label smoothing to label relaxation.
\newblock In \emph{Proceedings of the AAAI conference on artificial intelligence}, volume~35, 8583--8591.

\bibitem[{Loper et~al.(2023)Loper, Mahmood, Romero, Pons-Moll, and Black}]{loper2023smpl}
Loper, M.; Mahmood, N.; Romero, J.; Pons-Moll, G.; and Black, M.~J. 2023.
\newblock SMPL: A skinned multi-person linear model.
\newblock In \emph{Seminal Graphics Papers: Pushing the Boundaries, Volume 2}, 851--866.

\bibitem[{Lukasik et~al.(2020)Lukasik, Bhojanapalli, Menon, and Kumar}]{lukasik2020does}
Lukasik, M.; Bhojanapalli, S.; Menon, A.; and Kumar, S. 2020.
\newblock Does label smoothing mitigate label noise?
\newblock In \emph{International Conference on Machine Learning}, 6448--6458. PMLR.

\bibitem[{Mahmood et~al.(2019)Mahmood, Ghorbani, Troje, Pons-Moll, and Black}]{mahmood2019amass}
Mahmood, N.; Ghorbani, N.; Troje, N.~F.; Pons-Moll, G.; and Black, M.~J. 2019.
\newblock AMASS: Archive of motion capture as surface shapes.
\newblock In \emph{Proceedings of the IEEE/CVF international conference on computer vision}, 5442--5451.

\bibitem[{Maurice et~al.(2019)Maurice, Malais{\'e}, Amiot, Paris, Richard, Rochel, and Ivaldi}]{maurice2019human}
Maurice, P.; Malais{\'e}, A.; Amiot, C.; Paris, N.; Richard, G.-J.; Rochel, O.; and Ivaldi, S. 2019.
\newblock Human movement and ergonomics: An industry-oriented dataset for collaborative robotics.
\newblock \emph{The International Journal of Robotics Research}, 38(14): 1529--1537.

\bibitem[{Menache(2000)}]{menache2000understanding}
Menache, A. 2000.
\newblock \emph{Understanding motion capture for computer animation and video games}.
\newblock Morgan kaufmann.

\bibitem[{Mousavi~Hondori and Khademi(2014)}]{mousavi2014review}
Mousavi~Hondori, H.; and Khademi, M. 2014.
\newblock A review on technical and clinical impact of microsoft kinect on physical therapy and rehabilitation.
\newblock \emph{Journal of medical engineering}, 2014(1): 846514.

\bibitem[{M{\"u}ller, Kornblith, and Hinton(2019)}]{muller2019does}
M{\"u}ller, R.; Kornblith, S.; and Hinton, G.~E. 2019.
\newblock When does label smoothing help?
\newblock \emph{Advances in neural information processing systems}, 32.

\bibitem[{Noitom(2017)}]{noitom2017perception}
Noitom. 2017.
\newblock Perception neuron.

\bibitem[{Palermo et~al.(2022)Palermo, Cerqueira, Andr{\'e}, Pereira, and Santos}]{palermo2022raw}
Palermo, M.; Cerqueira, S.~M.; Andr{\'e}, J.; Pereira, A.; and Santos, C.~P. 2022.
\newblock From raw measurements to human pose-a dataset with low-cost and high-end inertial-magnetic sensor data.
\newblock \emph{Scientific Data}, 9(1): 591.

\bibitem[{Pereyra et~al.(2017)Pereyra, Tucker, Chorowski, Kaiser, and Hinton}]{pereyra2017regularizing}
Pereyra, G.; Tucker, G.; Chorowski, J.; Kaiser, {\L}.; and Hinton, G. 2017.
\newblock Regularizing neural networks by penalizing confident output distributions.
\newblock \emph{arXiv preprint arXiv:1701.06548}.

\bibitem[{Perlin(1985)}]{perlin1985image}
Perlin, K. 1985.
\newblock An image synthesizer.
\newblock \emph{ACM Siggraph Computer Graphics}, 19(3): 287--296.

\bibitem[{Real et~al.(2019)Real, Aggarwal, Huang, and Le}]{real2019regularized}
Real, E.; Aggarwal, A.; Huang, Y.; and Le, Q.~V. 2019.
\newblock Regularized evolution for image classifier architecture search.
\newblock In \emph{Proceedings of the aaai conference on artificial intelligence}, volume~33, 4780--4789.

\bibitem[{Shao et~al.(2025)Shao, Yi, Yin, Guo, Yong, and Xu}]{shao2025magshield}
Shao, Y.; Yi, X.; Yin, L.; Guo, S.; Yong, J.; and Xu, F. 2025.
\newblock MagShield: Towards Better Robustness in Sparse Inertial Motion Capture Under Magnetic Disturbances.
\newblock \emph{arXiv preprint arXiv:2506.22907}.

\bibitem[{Szegedy et~al.(2016)Szegedy, Vanhoucke, Ioffe, Shlens, and Wojna}]{szegedy2016rethinking}
Szegedy, C.; Vanhoucke, V.; Ioffe, S.; Shlens, J.; and Wojna, Z. 2016.
\newblock Rethinking the inception architecture for computer vision.
\newblock In \emph{Proceedings of the IEEE conference on computer vision and pattern recognition}, 2818--2826.

\bibitem[{Trumble et~al.(2017)Trumble, Gilbert, Malleson, Hilton, and Collomosse}]{trumble2017total}
Trumble, M.; Gilbert, A.; Malleson, C.; Hilton, A.; and Collomosse, J. 2017.
\newblock Total capture: 3d human pose estimation fusing video and inertial sensors.
\newblock In \emph{Proceedings of 28th British Machine Vision Conference}, 1--13.

\bibitem[{Vaswani et~al.(2017)Vaswani, Shazeer, Parmar, Uszkoreit, Jones, Gomez, Kaiser, and Polosukhin}]{vaswani2017attention}
Vaswani, A.; Shazeer, N.; Parmar, N.; Uszkoreit, J.; Jones, L.; Gomez, A.~N.; Kaiser, {\L}.; and Polosukhin, I. 2017.
\newblock Attention is all you need.
\newblock \emph{Advances in neural information processing systems}, 30.

\bibitem[{Von~Marcard et~al.(2017)Von~Marcard, Rosenhahn, Black, and Pons-Moll}]{von2017sparse}
Von~Marcard, T.; Rosenhahn, B.; Black, M.~J.; and Pons-Moll, G. 2017.
\newblock Sparse inertial poser: Automatic 3d human pose estimation from sparse imus.
\newblock In \emph{Computer graphics forum}, volume~36, 349--360. Wiley Online Library.

\bibitem[{Wei et~al.(2021)Wei, Liu, Liu, Niu, Sugiyama, and Liu}]{wei2021smooth}
Wei, J.; Liu, H.; Liu, T.; Niu, G.; Sugiyama, M.; and Liu, Y. 2021.
\newblock To smooth or not? when label smoothing meets noisy labels.
\newblock \emph{arXiv preprint arXiv:2106.04149}.

\bibitem[{Wu et~al.(2024)Wu, Yin, Guo, Qin et~al.}]{wu2024accurate}
Wu, Y.; Yin, L.; Guo, S.; Qin, Y.; et~al. 2024.
\newblock Accurate and steady inertial pose estimation through sequence structure learning and modulation.
\newblock \emph{Advances in Neural Information Processing Systems}, 37: 42468--42493.

\bibitem[{Yi, Pan, and Xu(2025)}]{yi2025improving}
Yi, X.; Pan, S.; and Xu, F. 2025.
\newblock Improving Global Motion Estimation in Sparse IMU-based Motion Capture with Physics.
\newblock \emph{arXiv preprint arXiv:2505.05010}.

\bibitem[{Yi et~al.(2022)Yi, Zhou, Habermann, Shimada, Golyanik, Theobalt, and Xu}]{yi2022physical}
Yi, X.; Zhou, Y.; Habermann, M.; Shimada, S.; Golyanik, V.; Theobalt, C.; and Xu, F. 2022.
\newblock Physical inertial poser (pip): Physics-aware real-time human motion tracking from sparse inertial sensors.
\newblock In \emph{Proceedings of the IEEE/CVF Conference on Computer Vision and Pattern Recognition}, 13167--13178.

\bibitem[{Yi, Zhou, and Xu(2021)}]{yi2021transpose}
Yi, X.; Zhou, Y.; and Xu, F. 2021.
\newblock Transpose: Real-time 3d human translation and pose estimation with six inertial sensors.
\newblock \emph{ACM Transactions on Graphics (TOG)}, 40(4): 1--13.

\bibitem[{Yi, Zhou, and Xu(2024)}]{yi2024pnp}
Yi, X.; Zhou, Y.; and Xu, F. 2024.
\newblock Physical Non-inertial Poser (PNP): Modeling Non-inertial Effects in Sparse-inertial Human Motion Capture.
\newblock In \emph{SIGGRAPH 2024 Conference Papers}.

\bibitem[{Yuan et~al.(2020)Yuan, Tay, Li, Wang, and Feng}]{yuan2020revisiting}
Yuan, L.; Tay, F.~E.; Li, G.; Wang, T.; and Feng, J. 2020.
\newblock Revisiting knowledge distillation via label smoothing regularization.
\newblock In \emph{Proceedings of the IEEE/CVF conference on computer vision and pattern recognition}, 3903--3911.

\bibitem[{Zhang et~al.(2021)Zhang, Jiang, Hou, Wei, Han, Li, and Cheng}]{zhang2021delving}
Zhang, C.-B.; Jiang, P.-T.; Hou, Q.; Wei, Y.; Han, Q.; Li, Z.; and Cheng, M.-M. 2021.
\newblock Delving deep into label smoothing.
\newblock \emph{IEEE Transactions on Image Processing}, 30: 5984--5996.

\bibitem[{Zhou et~al.(2021)Zhou, Song, Chen, Zhou, Wang, Yuan, and Zhang}]{zhou2021rethinking}
Zhou, H.; Song, L.; Chen, J.; Zhou, Y.; Wang, G.; Yuan, J.; and Zhang, Q. 2021.
\newblock Rethinking soft labels for knowledge distillation: A bias-variance tradeoff perspective.
\newblock \emph{arXiv preprint arXiv:2102.00650}.

\bibitem[{Zhou et~al.(2019)Zhou, Barnes, Lu, Yang, and Li}]{zhou2019continuity}
Zhou, Y.; Barnes, C.; Lu, J.; Yang, J.; and Li, H. 2019.
\newblock On the continuity of rotation representations in neural networks.
\newblock In \emph{Proceedings of the IEEE/CVF conference on computer vision and pattern recognition}, 5745--5753.

\bibitem[{Zhou et~al.(2025)Zhou, Wei, Zhang, Dou, Qu, and Cai}]{zhou2025training}
Zhou, Z.; Wei, S.; Zhang, X.; Dou, W.; Qu, M.; and Cai, Y. 2025.
\newblock Training Deep Neural Networks with Virtual Smoothing Classes.
\newblock In \emph{Proceedings of the AAAI Conference on Artificial Intelligence}, volume~39, 23036--23044.

\bibitem[{Zuo et~al.(2025)Zuo, Huang, Jiang, Yao, Shi, Cao, Yi, Xu, Guo, and Qin}]{zuo2025transformer}
Zuo, C.; Huang, J.; Jiang, X.; Yao, Y.; Shi, X.; Cao, R.; Yi, X.; Xu, F.; Guo, S.; and Qin, Y. 2025.
\newblock Transformer IMU calibrator: Dynamic on-body IMU calibration for inertial motion capture.
\newblock \emph{arXiv preprint arXiv:2506.10580}.

\end{thebibliography}


\appendix
\renewcommand{\thesection}{\Alph{section}}
\begin{center}{\LARGE\bf Appendix\par}\vspace{1em}\end{center}
\addcontentsline{toc}{section}{Appendix}

\section{Proofs of Naive Entropy-Enhancement Strategies}
In this section, we demonstrate that naive entropy-enhancement strategies, including blending Gaussian noise or uniform noise, are ineffective as they cannot ensure a sufficiently large noise amplitude for effective regularization without disrupting motion properties.

\begin{proposition}
    Applying Gaussian noise or uniform noise with sufficiently large noise amplitude does not satisfy Property 1 (Temporal Smoothness).
\end{proposition}
\begin{proof}
When perturbing labels with a noise vector $u$, for time $ t $, the perturbed label is given by:
\begin{equation}
    R'(t) = (1-\epsilon)R(t) + \epsilon u(t)
    \label{eq:motion_label_smoothing}
\end{equation}
According to Property 1, the ground truth labels naturally exhibit temporal smoothness, as characterized by the constraint $ \| \omega(t)\|_2 = \|\Delta R(t) / \Delta t\|_2 \leq M $.
Thus, we have:
\begin{equation}
\begin{aligned}
    &\Delta R'(t) &= &(1-\epsilon)\Delta R(t) &+ &\epsilon \Delta u(t)\\
    &\|\Delta R'(t)\|_2 &\leq &(1-\epsilon)\|\Delta R(t)\|_2 &+ &\epsilon \|\Delta u(t)\|_2\\
    &\|\Delta R'(t) /\Delta t \|_2 &\leq &(1-\epsilon)\|\Delta  R(t) /\Delta t \|_2 &+ &\epsilon \|\Delta u(t) /\Delta t \|_2\\
    &\|R'(t) /\Delta t \|_2 &\leq &(1-\epsilon)M &+ &\epsilon \|\Delta u(t) /\Delta t \|_2
\end{aligned}
\end{equation}
To ensure that $\| \omega'(t)\|_2 = \|R'(t) /\Delta t \|_2  \leq M$ (i.e., satisfying Property 1),
\begin{equation}
    \begin{aligned}
        &\|R'(t) /\Delta t \|_2 \leq (1-\epsilon)M + \epsilon \|\Delta u(t) /\Delta t \|_2 \leq M\\
        &\Leftrightarrow \epsilon \|\Delta u(t) /\Delta t \|_2 \leq M - (1-\epsilon)M\\
        &\Leftrightarrow \|\Delta u(t) /\Delta t \|_2 \leq M
    \end{aligned}
    \label{eq:noise_temporal_contraint}
\end{equation}
which requires the gradient magnitude of the noise term $ u(t) $ to be bounded by the same physiological upper bound $ M$.

\vspace{2mm}
\noindent \textit{Naive entropy-Enhancement strategies.} 
Applying Eq.~\ref{eq:noise_temporal_contraint} to naive entropy-Enhancement strategies, we have:
\begin{itemize}
    \item For Gaussian noise $ u_{gau} \sim \mathcal{N}(\mu, \sigma^2) $, given two Gaussian samples $X,Y \sim \mathcal{N}(\mu, \sigma^2)$, we have:
    \begin{equation}
        \|\Delta u_{gau}(t) /\Delta t \|_2 \leq \|(X-Y)/\Delta t\|_2 \leq M
    \end{equation}
    Since $X-Y \sim \mathcal{N}(\mu-\mu, \sigma^2+\sigma^2)$, it is unbounded and the chance to violate Property 1 ($\|\Delta u_{gau}(t) /\Delta t \|_2>M$) increases with increasing noise amplitude determined by $\sigma$.
    \item For uniform noise $ u_{uni} \sim \mathcal{U}(a, b) $, given two samples $X,Y \sim \mathcal{U}(a, b)$, we have:
    \begin{equation}
        \|\Delta u_{uni}(t) /\Delta t \|_2 \leq \|(X-Y)/\Delta t\|_2 \leq M
    \end{equation}
    where the mean of $X-Y$ is 0 with variance $2\times(b-a)^2/12$. It can be observed that the chance to violate Property 1 ($\|\Delta u_{uni}(t) /\Delta t \|_2>M$) increases with the noise amplitude determined by $b-a$.
\end{itemize}
Therefore, applying Gaussian noise or uniform noise are both suboptimal as increasing their noise amplitudes risks disrupting Property 1.

\end{proof}

\begin{proposition}
    Applying Gaussian noise or uniform noise with sufficiently large noise amplitude does not satisfy Property 2 (Joint Correlation).
\end{proposition}
\begin{proof}
According to Property 2, the ground truth motion label $R_{\text{child}}$ satisfy:
\begin{equation}
\begin{aligned}
\phi_{min}(R_{\text{parent}}(t)) \leq  R_{\text{child}}(t) \leq \phi_{max}(R_{\text{parent}}(t))
\label{eq:prop2_condition}
\end{aligned}
\end{equation}
Same as Eq.~\ref{eq:motion_label_smoothing}, after applying noise, the perturbed motion labels become:
\begin{equation}
\begin{aligned}
    & R'_{\text{parent}}(t) = (1-\epsilon)R_{\text{parent}}(t) + \epsilon u_{\text{parent}}(t), \\
    & R'_{\text{child}}(t) = (1-\epsilon)R_{\text{child}}(t) + \epsilon u_{\text{child}}(t)
    \label{eq:parent_child_pertubation}
\end{aligned}
\end{equation}
And Property 2 requires the child joint angle $R'_{\text{child}}$:
\begin{equation}
\begin{aligned}
\phi_{min}(R'_{\text{parent}}(t)) \leq  R'_{\text{child}}(t) \leq \phi_{max}(R'_{\text{parent}}(t))
\end{aligned}
\label{eq:prop.2}
\end{equation}
Since $\phi$ is a Lipschitz function, let its Lipschitz constant be $\mathrm{L}_\phi$, we have:
\begin{equation}
\begin{aligned}
    R'_{\text{child}}(t) \geq \phi_{min}(R_{\text{parent}}(t)) - \mathrm{L}_\phi(R'_{\text{parent}}(t) - R_{\text{parent}}(t)) \\
    R'_{\text{child}}(t) \leq \phi_{max}(R_{\text{parent}}(t)) + \mathrm{L}_\phi(R'_{\text{parent}}(t) - R_{\text{parent}}(t)) 
\end{aligned}
\end{equation}
Substituting Eq.~\ref{eq:parent_child_pertubation}, we have
\begin{equation}
\begin{aligned}
    R'_{\text{child}}(t) \geq \phi_{min}(R_{\text{parent}}(t)) - \epsilon \mathrm{L}_\phi(R_{\text{parent}}(t) - u_{\text{parent}}(t)) \\
    R'_{\text{child}}(t) \leq \phi_{max}(R_{\text{parent}}(t)) + \epsilon \mathrm{L}_\phi(R_{\text{parent}}(t) - u_{\text{parent}}(t))
\end{aligned}
\end{equation}
Without loss of generality, we take the lower bound equation as an example and substitute Eq.~\ref{eq:parent_child_pertubation}, and have:
\begin{equation}
\begin{aligned}
    (1-\epsilon)R_{\text{child}}(t) &+ \epsilon u_{\text{child}}(t) \\ 
    &\geq \phi_{min}(R_{\text{parent}}(t)) \\
    &  - \epsilon \mathrm{L}_\phi(R_{\text{parent}}(t) - u_{\text{parent}}(t)) 
\end{aligned}
\end{equation}
which is equivalent to:
\begin{equation}
\begin{aligned}
    &R_{\text{child}}(t) \geq \frac{\phi_{min}(R_{\text{parent}}(t)) - \epsilon \mathrm{L}_\phi \cdot R_{\text{parent}}(t)}{1-\epsilon}\\
    &+ \frac{\epsilon \mathrm{L}_\phi\cdot u_{\text{parent}}(t) - \epsilon u_{\text{child}}(t)}{1-\epsilon}
\end{aligned}
\end{equation}
Recall the condition in Eq.~\ref{eq:prop2_condition}, the above requirement is guaranteed to be satisfied when:
\begin{equation}
\begin{aligned}
    &\phi_{min}(R_{\text{parent}}(t)) \geq \frac{\phi_{min}(R_{\text{parent}}(t)) - \epsilon\mathrm{L}_\phi \cdot R_{\text{parent}}(t)}{1-\epsilon}\\
    &+ \frac{\epsilon (\mathrm{L}_\phi\cdot u_{\text{parent}}(t) - u_{\text{child}}(t))}{1-\epsilon}
\end{aligned}
\end{equation}
which is equivalent to:
\begin{equation}
\begin{aligned}
    \phi_{min}(R_{\text{parent}}(t)) \leq \mathrm{L}_\phi \cdot R_{\text{parent}}(t) - \mathrm{L}_\phi\cdot u_{\text{parent}}(t) - u_{\text{child}}(t)
\end{aligned}
\end{equation}
Thus, we have:
\begin{equation}
\begin{aligned}
    \mathrm{L}_\phi\cdot u_{\text{parent}}(t) + u_{\text{child}}(t) \leq \mathrm{L}_\phi \cdot R_{\text{parent}}(t) - \phi_{min}(R_{\text{parent}}(t))
\end{aligned}
\label{eq:prop2_resulting_relationship}
\end{equation}
which indicates that $\mathrm{L}_\phi\cdot u_{\text{parent}}(t) + u_{\text{child}}(t)$ should be bounded by a constant. Same logic applies for the lower bound.

\noindent \textit{Naive entropy-Enhancement strategies.} 
Similar to the proof in Proposition 1, we have:
\begin{itemize}
    \item For Gaussian noise $\mathcal{N}(\mu, \sigma^2) $, we have $u_{\text{parent}},u_{\text{child}} \sim \mathcal{N}(\mu, \sigma^2)$, and $\mathrm{L}_\phi\cdot u_{\text{parent}}(t) + u_{\text{child}}(t) \sim \mathcal{N}(\mathrm{L}_\phi\mu+\mu, \mathrm{L}_\phi\sigma^2+\sigma^2)$, it is unbounded and the chance to violate Property 2 (Eq.~\ref{eq:prop2_resulting_relationship}) increases with the noise amplitude determined by $\mu$ and $\sigma$.
    \item For uniform noise $ \mathcal{U}(a, b) $, we have $u_{\text{parent}},u_{\text{child}} \sim \mathcal{U}(a, b)$, and the mean of $\mathrm{L}_\phi\cdot u_{\text{parent}}(t) + u_{\text{child}}(t)$ is $a+b$ with variance $(b-a)^2/6$. It can be observed that the chance to violate Property increases with the noise amplitude determined by $a+b$ and $b-a$.
\end{itemize}
Therefore, applying Gaussian noise or uniform noise are both suboptimal as increasing their noise amplitudes risks disrupting Property 2.

\end{proof}


\begin{proposition}
    Applying Gaussian noise or uniform noise does not satisfy Property 3 ( Low-frequency Dominance).
\end{proposition}
\begin{proof}
For Gaussian noise and uniform noise, their respective noise models are defined as:
$ u_{gau} \sim \mathcal{N}(\mu, \sigma^2) $,
$ u_{uni} \sim \mathcal{U}(a, b) $, with independence across all time instances.
Since both Gaussian noise and uniform noise are white noise, their power spectral density (PSD) $P_c$ is {\it constant}:
\begin{equation}
\begin{aligned} 
        P_c^{u}(f) \propto {\sigma^2}{\Delta t},
\end{aligned}
\end{equation}
where $\sigma^2$ is the variance of the distribution ($\frac{(a-b)^2}{12}$ for $u_{uni}$). A constant PSD obviously violates Property 3.


\end{proof}




\section{Proofs of Our Skeleton-based Perlin Noise}
In this section, we demonstrate that our skeleton-based Perlin noise is compatible with the three motion properties while perserving a sufficiently large noise amplitude for effective regularization, noting that ground-truth human motion data inherently satisfy these properties.
\begin{proposition}
Our skeleton-based Perlin noise (sk-Perlin) satisfies Property 1 (Temporal Smoothness) with a sufficiently large noise amplitude to enable effective regularization.
\end{proposition}
\begin{proof}


Same as Eq.~\ref{eq:noise_temporal_contraint} in Proposition 1, to satisfy Property 1, we require our skeleton-based Perlin noise function $u(t)$ to satisfy:
\begin{equation}
    \|\Delta u(t) /\Delta t \|_2 \leq M
\end{equation}
Specifically, our $u(t)$ is defined as:
\begin{equation}
\begin{aligned}
    u(t) = \sum_{k=0}^{oct-1} A_k \cdot u_b(f_k \cdot t)
\end{aligned}
\label{eq:perlin}
\end{equation}
where $A_k = S_b \cdot p^k$ is the amplitude of octave $k$, $f_k = S_t \cdot l^k$ is the frequency of octave $k$, $u_b(\cdot)$ is the base Perlin noise function.
Therefore, we have:
\begin{equation}
\frac{\Delta u}{\Delta t} = \sum_{k=0}^{oct-1} A_k \cdot f_k \cdot \frac{\Delta u_b(f_k \cdot t)}{\Delta t}
\end{equation}
Since the base Perlin noise $u_b$ has a bounded gradient magnitude:
\begin{equation}
\left\| \frac{\Delta u_b}{\Delta t} \right\| \leq G_{\max} \quad \forall x \in \mathbb{R}
\end{equation}
where $G_{\max}$ is the maximum gradient magnitude (theoretical value $\approx 1.0$), we have:
\begin{equation}
\begin{aligned}
\left\| \frac{\Delta u}{\Delta t} \right\| 
&\leq \sum_{k=0}^{oct-1} |A_k| \cdot |f_k| \cdot G_{\max} 
\end{aligned}
\end{equation}
Substitute $A_k$ and $f_k$ expressions, we have:
\begin{equation}
\begin{aligned}
\left\| \frac{\Delta u}{\Delta t} \right\| &\leq G_{\max} \sum_{k=0}^{oct-1} (S_b \cdot p^k) \cdot (S_t \cdot l^k) \\
&= G_{\max} \cdot S_b \cdot S_t \sum_{k=0}^{oct-1} (p \cdot l)^k 
\end{aligned}
\end{equation}
Let $r = p \cdot l$, the sum is a geometric series:
\begin{equation}
\sum_{k=0}^{oct-1} r^k = \begin{cases} 
\frac{1 - r^{oct}}{1 - r} & \text{if } r \neq 1 \\
oct & \text{if } r = 1 
\end{cases}
\end{equation}
Thus, we have:
\begin{equation}
\begin{aligned}
\left\| \frac{\Delta u}{\Delta t} \right\| \leq G_{\max} \cdot S_b \cdot S_t \cdot S(r, oct)\\
\text{where}~S(r, n) = \begin{cases} 
\frac{1 - r^n}{1 - r} & r \neq 1 \\
n & r = 1    
\end{cases}
\end{aligned}
\end{equation}
As a result, Property 1 is guaranteed to be satisfied if:
\begin{equation}
    G_{\max} \cdot S_b \cdot S_t \cdot S(r, oct) \leq M
\end{equation}
Since the condition is not solely determined by the noise amplitude $S_b$, our skeleton-based Perlin noise can easily satisfy Property 1 by selecting a sufficiently large $S_b$ and a compensating small $S_t$ and/or $S(r, oct)$.

\end{proof}

\begin{proposition}
Our skeleton-based Perlin noise (sk-Perlin) satisfies Property 2 (Joint Correlation) with a sufficiently large noise amplitude to enable effective regularization.
\end{proposition}
\begin{proof}
Recalling Eq. 15, $\mathrm{L}_\phi\cdot u_{\text{parent}}(t) + u_{\text{child}}(t)$ should be bounded by a constant. For our Perlin noise, we have
\begin{equation}
\begin{aligned}
&\|u_{parent}\|_2 \leq B, \\
&\|u_{child} - u_{parent}\|_2 \leq C;
\end{aligned}
\end{equation}
where $B$ is the global noise amplitude upper bound of perlin noise (Eq. 18):
\begin{equation}
    B = S_b \times \left( \sum_{i=0}^{oct-1} p^i \right) \times \sqrt{d},
\end{equation}
$ d = 6 $ is the rotation dimension in $ \mathbb{R}^6 $, and $ C $ is the upper bound on inter-joint noise variation:
\begin{equation}
    C = K_g \times \| d \|_2 \times \frac{g}{S_s \cdot f},
\end{equation}
where $ K_g \approx 2.5 $ is the Perlin gradient constant, $ \| d \|_2 $ is the inter-joint anatomical distance (e.g., shoulder-elbow distance), $ f $ is the spatial frequency (Hz per unit distance), and $ g $ is the gradient weight as described in the Method section of the main paper. The biomechanical implication is that $ \frac{C}{\|d\|_2} $ represents the noise difference per unit distance. 

For our upper bound, 
\begin{equation}
\|\mathrm{L}_\phi\cdot u_{\text{parent}}(t) + u_{\text{child}}(t)\|_2 \leq (\mathrm{L}_\phi + 1)B + C
\end{equation}

Since B and C is defined by the parameters of perlin noise, Thus, since $ \mathrm{L}_\phi + 1 $ is a constant, to keep the cumulative error sufficiently small, even with a larger noise intensity $ S_b $, $ B $ and $ C $ can be reduced by decreasing $ oct $ or increasing $ S_s $. Consequently, Property 2 can be satisfied while using a larger $ S_b $, completing the proof.
\end{proof}

\begin{table*}[ht]
    \centering
    \small
    \begin{tabular*}{\textwidth}{lcccccccc}
        \toprule
         & \multicolumn{4}{c}{\textbf{TotalCapture}} & \multicolumn{4}{c}{\textbf{ANDY}} \\
        \cmidrule(lr){2-5} \cmidrule(lr){6-9}
        \textbf{Method} & \textbf{SIP Err} & \textbf{Angular Err} & \textbf{Joint Err} & \textbf{Mesh Err} & \textbf{SIP Err} & \textbf{Angular Err} & \textbf{Joint Err} & \textbf{Mesh Err} \\
        \midrule
        \textit{TransPose} & 14.28±6.90 & 11.36±4.83 & 5.31±3.02 & 5.89±3.29 & 31.91±18.36 & 39.55±27.83 & 13.77±10.54 & 18.96±14.40 \\
        \textit{TP+Ours} & 12.49±7.41 & 9.43±4.63 & 5.00±3.06 & 5.55±3.36 & 31.20±18.89 & 37.54±26.45 & 13.42±10.50 & 18.46±14.26 \\
        \midrule
        \textit{PIP} & 11.16±5.50 & 10.71±4.64 & 4.55±2.58 & 5.26±2.93 & 29.59±19.22 & 36.82±27.66 & 13.56±11.14 & 18.79±15.47 \\
        \textit{PIP+Ours} & 10.54±5.33 & 8.97±4.21 & 4.38±2.59 & 5.07±2.95 & 29.04±18.81 & 34.16±25.08 & 13.49±11.08 & 18.67±15.35 \\
        \midrule
        \textit{GlobalPose} & 9.85±4.60 & 9.55±4.13 & 3.96±1.94 & 4.35±2.17 & 39.58±25.19 & 46.32±31.76 & 17.66±14.03 & 23.45±18.64 \\
        \textit{GP+Ours} & 7.84±4.50 & 7.87±3.40 & 3.26±1.80 & 3.75±2.00 & 36.71±23.43 & 40.68±28.60 & 16.92±13.42 & 22.22±17.41 \\
        \midrule
        & \multicolumn{4}{c}{\textbf{CIP}} & \multicolumn{4}{c}{\textbf{DIP-IMU}} \\
        \cmidrule(lr){2-5} \cmidrule(lr){6-9}
        \textbf{Method} & \textbf{SIP Err} & \textbf{Angular Err} & \textbf{Joint Err} & \textbf{Mesh Err} & \textbf{SIP Err} & \textbf{Angular Err} & \textbf{Joint Err} & \textbf{Mesh Err} \\
        \midrule
        \textit{TransPose} & 28.46±17.45 & 20.09±11.66 & 10.82±7.09 & 11.91±7.54 & 14.04±6.67 & 8.39±4.29 & 4.86±2.60 & 5.80±3.03 \\
        \textit{TP+Ours} & 26.19±17.30 & 19.02±11.34 & 10.65±7.12 & 11.71±7.56 & 13.57±6.56 & 8.06±4.04 & 4.64±2.44 & 5.50±2.81 \\
        \midrule
        \textit{PIP} & 25.57±13.37 & 19.98±10.76 & 9.02±5.61 & 10.76±6.54 & 12.08±5.77 & 8.18±4.29 & 4.33±2.36 & 5.06±2.70 \\
        \textit{PIP+Ours} & 23.87±11.93 & 19.04±10.10 & 8.60±5.54 & 10.27±6.48 & 11.62±5.46 & 8.14±4.14 & 4.18±2.25 & 4.88±2.57 \\
        \midrule
        \textit{GlobalPose} & 23.04±11.17 & 15.32±9.08 & 6.98±4.18 & 8.13±4.72 & 13.77±6.40 & 8.09±4.11 & 4.36±2.55 & 5.08±2.93 \\
        \textit{GP+Ours} & 22.32±10.73 & 14.36±8.74 & 6.47±4.12 & 7.70±4.65 & 13.50±6.31 & 8.14±4.02 & 4.27±2.44 & 4.98±2.83 \\
        \bottomrule
    \end{tabular*}
    \normalsize
    \setlength{\tabcolsep}{6pt}
    \caption{Quantitative comparisons between baseline methods and their enhanced versions incorporating our motion label smoothing method, with additional Angular Error metric and standard deviations.}
    \label{tab:comparison}
\end{table*}

\begin{table*}[t]
    \centering
    \begin{tabular}{lcccclcccc}
        \toprule
        & \multicolumn{4}{c}{\textbf{Base Scale $S_b$}} & & \multicolumn{4}{c}{\textbf{Time Scale $S_t$}} \\
        \cmidrule(lr){2-5} \cmidrule(lr){7-10}
        \textbf{Value} & \textbf{SIP Err} & \textbf{Angular Err} & \textbf{Joint Err} & \textbf{Mesh Err} & \textbf{Value} & \textbf{SIP Err} & \textbf{Angular Err} & \textbf{Joint Err} & \textbf{Mesh Err} \\
        \midrule
        0.03 & 8.38±4.51 & 8.20±3.34 & 3.62±1.98 & 4.24±2.22 & 0.1 & 8.22±4.40 & 8.04±3.27 & 3.56±1.94 & 4.12±2.16 \\
        0.05 & 8.24±4.44 & 8.07±3.30 & 3.56±1.94 & \textbf{4.13±2.16} & 0.3 & \textbf{8.21±4.38} & 8.01±3.26 & \textbf{3.54±1.93} & 4.12±2.16\\
        \underline{0.07} & 8.25±4.46 & \textbf{8.03±3.30} & \textbf{3.57±1.96} & 4.14±2.19 & \underline{0.5} & 8.26±4.41 & \textbf{8.00±3.28} & 	\textbf{3.54±1.93} & \textbf{4.11±2.15} \\
        0.09 & \textbf{8.23±4.38} & 8.05±3.30 & 3.57±1.96 & 4.16±2.19 & 0.7 & 8.24±4.42 & 8.04±3.29 & 3.55±1.95 & 4.12±2.17\\
        0.11 & 8.32±4.44 & 8.13±3.31 & 3.59±1.94 & 4.18±2.17 & 0.9 & 8.32±4.40 & 8.03±3.29 & 3.57±1.96 & 4.14±2.18\\
        \midrule
        & \multicolumn{4}{c}{\textbf{Space Scale $S_s$}} & & \multicolumn{4}{c}{\textbf{Octaves $oct$}} \\
        \cmidrule(lr){2-5} \cmidrule(lr){7-10}
        \textbf{Value} & \textbf{SIP Err} & \textbf{Angular Err} & \textbf{Joint Err} & \textbf{Mesh Err} & \textbf{Value} & \textbf{SIP Err} & \textbf{Angular Err} & \textbf{Joint Err} & \textbf{Mesh Err} \\
        \midrule
        0.1 & 8.31±4.45 & 8.09±3.34 & 3.58±1.97 & 4.17±2.21 & 2 & 8.41±4.46 & 8.22±3.35 & 3.62±1.97 & 4.20±2.19 \\
        0.3 & 8.38±4.46 & 8.11±3.32 & 3.58±1.96 & 4.16±2.18& 3 & 8.21±4.45 & 8.00±3.31 & 3.55±1.95 & 4.12±2.18 \\
        0.5 & \textbf{8.21±4.42} & 8.09±3.31 & 3.58±1.95 & 4.16±2.19 & 4 & \textbf{8.17±4.43} & 7.99±3.31 & 3.55±1.96 & 4.11±2.18 \\
        \underline{0.7} & 8.30±4.44 & \textbf{8.01±3.28} & \textbf{3.56±1.94} & 4.15±2.18 & \underline{5} & 8.18±4.36 & \textbf{7.95±3.27} & \textbf{3.50±1.92} & \textbf{4.07±2.14} \\
        0.9 & 8.32±4.45 & 8.02±3.28 & 3.56±1.95 & \textbf{4.14±2.18}
         & 6 & 8.26±4.39 & 8.04±3.27 & 3.54±1.94 & 4.11±2.16 \\
        \midrule
        & \multicolumn{4}{c}{\textbf{Persistence $p$}} & & \multicolumn{4}{c}{\textbf{Lacunarity $l$}} \\
        \cmidrule(lr){2-5} \cmidrule(lr){7-10}
        \textbf{Value} & \textbf{SIP Err} & \textbf{Angular Err} & \textbf{Joint Err} & \textbf{Mesh Err} & \textbf{Value} & \textbf{SIP Err} & \textbf{Angular Err} & \textbf{Joint Err} & \textbf{Mesh Err} \\
        \midrule
        0.4 & 8.31±4.45 & 8.09±3.28 & 3.58±1.95 & 4.16±2.17& 1.0 & 8.15±4.39 & 8.05±3.27 & \textbf{3.53±1.94} & \textbf{4.09±2.16} \\
        \underline{0.5} & 8.27±4.44 & \textbf{7.98±3.30} & \textbf{3.56±1.95} & \textbf{4.15±2.17} & \underline{1.5} & \textbf{8.13±4.38} & 8.00±3.31 & \textbf{3.53±1.95} & \textbf{4.09±2.17} \\
        0.6 & \textbf{8.25±4.45} & 8.01±3.26 & 3.57±1.96 & 4.15±2.18 & 2.0 & 8.23±4.41 & 8.08±3.32 & 3.56±1.95 & 4.15±2.18 \\
        0.7 & 8.31±4.43 & 8.09±3.30 & 3.59±1.94 & 4.16±2.17 & 2.5 & 8.25±4.46 & 8.09±3.30 & 3.57±1.96 & 4.15±2.19 \\
        0.8 & 8.39±4.45 & 8.16±3.31 & 3.62±1.97 & 4.22±2.19& 3.0 & 8.18±4.39 & \textbf{7.95±3.28} & 3.54±1.94 & 4.10±2.16 \\
        \bottomrule
    \end{tabular}
    \caption{Quantitative results on hyperparameter selection using GlobalPose+Ours and the TotalCapture~\cite{trumble2017total} dataset. When varying a specific hyperparameter, all other hyperparameters are fixed at their respective optimal values. \underline{Underline}: the optimal choice. \textbf{Bold}: best results.}
    \label{tab:hyper}
\end{table*}

\begin{proposition}
    Our $\text{sk-Perlin}$ noise satisfies Property 3 (Low-frequency Dominance).
\end{proposition}
\begin{proof}
The power spectrum density (PSD) $P_c$ of a pink noise such as our skeleton-based Perlin noise $u$ (Eq.~\ref{eq:perlin}) follows a $1/f^\beta$ distribution ($\beta \approx 2$):
\begin{equation}
    P_c^{u}(f) \propto \frac{1}{f^{2}},
\end{equation}
where $f$ denotes the frequency of the signal. Obviously, our skeleton-based Perlin noise satisfies Property 3.



\end{proof}

\section{Implementation Details}
\subsection{System Hardware}
Our task involves real-time human pose estimation using six IMUs. In the live demonstration, we employ six Noitom PN Lab series sensors~\cite{noitom2017perception} operating at a frame rate of 60 fps. Training is conducted on an Nvidia RTX 4090 graphics card, while the live demo runs in real-time on a laptop equipped with an Intel(R) Core(TM) i7-12700H CPU, without GPU acceleration.

\subsection{Consent}
Live demo recordings and the use of images in our demo video and teaser figure were conducted with explicit written consent from all participants.

\section{Additional Experiments}

\subsubsection{Additional Metrics of Table 1.}
We show a more comprehensive set of metrics of Table 1 in the main paper in Table~\ref{tab:comparison}, including additional standard deviations and angular errors.

\subsubsection{Hyperparameter Selection.}
We validate the hyperparameter selection for our skeleton-based Perlin noise by comparing the quantitative results for the choices of each parameter, including base scale $S_b$, time scale $S_t$, space scale (joint scale) $S_s$, persistence $p$, octaves $oct$ and lacunarity $l$, as presented in Table~\ref{tab:hyper}. Accordingly, we select the parameter set yielding the best results as:
\begin{itemize}
    \item $ S_b = 0.07 $ determines the overall noise magnitude, which is sufficiently large for effective regularization while satisfying all properties.
    \item $ S_t = 0.5$ stretches the time axis to enhance temporal smoothness.
    \item $ S_s = 0.7 $ controls the noise variation between different dimensions.
    \item $ p = 0.5 $ attenuates high-frequency effects, reinforcing the dominance of low-frequency noise. 
    \item $ oct = 5 $ prevents the noise from becoming overly complex by reducing the number of high-frequency layers.
    \item $l = 1.5$ controls the frequency doubling coefficient for each octave layer.
\end{itemize}

\subsubsection{Byproduct of Global Translation Improvement.}

\begin{figure}[t]
	\centering
 \includegraphics[width=0.8\linewidth]{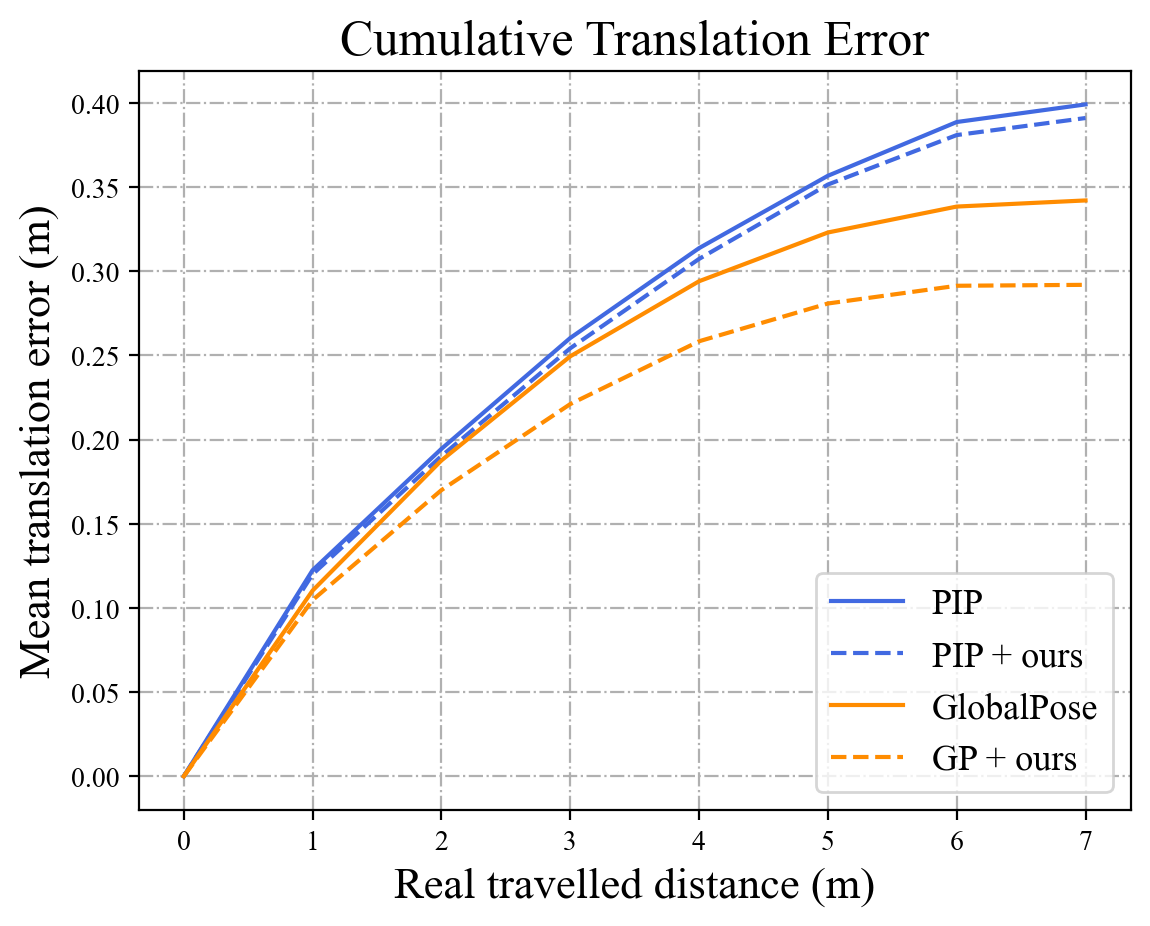}
	\caption{Comparison of translation drifting error on the TotalCapture dataset. We plot the global position error accumulation curve with respect to the ground truth traveled distance. A lower curve indicates smaller drift.
}
	\label{fig:trans}
\end{figure}

 Additionally, we observe that our method improves global translation as a byproduct by applying it to baseline methods based on physical optimization~\cite{yi2022physical}. In these algorithms, global translation is computed as an optimization problem, derived from information including local pose, joint velocities, and foot-ground contact probabilities. Since our method improves pose estimation, it also indirectly enhances the accuracy of global translation predictions (Fig.~\ref{fig:trans}). 
 
\subsubsection{Training from Scratch (Re-training).} To further validate the effectiveness of our motion label smoothing, in addition to fine-tuning the pose estimation network with smoothed labels, we also assess the impact of training from scratch (re-training). We adhere strictly to the training protocols of the baseline method, including utilizing the AMASS~\cite{mahmood2019amass} dataset and the DIP-IMU~\cite{huang2018deep} training split as the training set, with IMU data from the AMASS dataset synthesized using the data generation method proposed by \cite{yi2024pnp}.
As shown in Table~\ref{tab:retrain}, training the same network using our smoothed motional labels significantly outperforms those using the original motion labels, further demonstrating the superiority of our approach.


\section{Additional Discussion of Other Alternatives}
\subsubsection{Naive Temporal Smoothing.}
By directly appling Gaussian smoothing to labels, practically, is to perform one-dimensional Gaussian smoothing along the time axis for the R6D data of each joint. Intuitively, this naive approach smooths fluctuations along the time axis, reducing the uncertainty in the distribution. Therefore, this approach does not align with the proposed entropy-enhancement interpretation of label smoothing:

\begin{proposition}
    Gaussian smoothing is not an entropy-enhancement method.
\end{proposition}
\begin{proof}
Consider the Shannon entropy of the R6D component for each joint in the time series, defined as $H = - \sum_{i} p(x_i) \log p(x_i)$, where $ p(x_i) $ represents the probability density of the label $ x_i $ of frame $i$. After applying the filter, the output distribution $ p'(x) $ is the convolution of the original distribution $ p(x) $ with the Gaussian kernel $ G(x, \sigma) $:
\begin{equation}
    p'(x) = (p * G)(x) = \int p(x - \tau) G(\tau, \sigma) \, d\tau
    \nonumber
\end{equation}
where $ G(\tau, \sigma) = \frac{1}{\sqrt{2\pi\sigma^2}} e^{-\frac{\tau^2}{2\sigma^2}} $. According to information theory, the entropy of the output distribution $ p'(x) $ after convolving a distribution $ p(x) $ with a Gaussian kernel $ G(x, \sigma) $ satisfies 
\begin{equation}
    H[p'] \leq H[p] + H[G]
    \nonumber
\end{equation} 
where $ H[G] $ is the entropy of the Gaussian kernel, dependent on $ \sigma $. For a fixed $ \sigma $, $ H[G] $ remains a constant. More specifically, the convolution operation causes the distribution $ p'(x) $ to approach a Gaussian distribution (due to the central limit theorem's convolution effect), with the entropy of a Gaussian distribution given by $H_{\text{Gaussian}} = \frac{1}{2} \log (2\pi e \sigma^2)$. If the initial $ p(x) $ contains multiple peaks or irregular patterns, its entropy exceeds that of a single Gaussian distribution. Since the original label is highly unlikely to be a perfect Gaussian distribution, the entropy decreases after filtering.
\end{proof}

\begin{table}[ht]
    \centering
    \setlength{\tabcolsep}{4pt}
    \small
    \begin{tabular}{lcccc}
        \toprule
        \textbf{Method} & \textbf{SIP Err} & \textbf{Ang Err} & \textbf{Joint Err} & \textbf{Mesh Err}\\
        \midrule
        \textit{TransPose} & 14.22±6.60 & 11.50±5.00 & 5.39±2.81 & 5.98±3.08\\ 
        \textit{TP + Ours} & 13.74±6.77 & 11.44±4.08 & 5.38±2.18 & 5.96±3.08\\
        \midrule
        \textit{PIP} & 12.11±6.84 & 11.53±5.41 & 4.97±2.96 & 5.72±3.39\\ 
        \textit{PIP + Ours} & 11.77±6.82 & 11.40±5.39 & 4.81±2.96 & 5.58±3.40\\
        \midrule
        \textit{GlobalPose} & 10.49±4.73 & 9.91±3.90 & 4.27±2.14 & 4.80±2.38\\ 
        \textit{GP + Ours} & 9.08±4.56 & 8.60±3.47 & 3.95±2.06 & 4.55±2.30\\
        \toprule
    \end{tabular}
    \caption{
    Quantitative results on retraining (training-from-scratch) pose estimation models with the proposed motion label smoothing. Evaluation is done using the TotalCapture dataset.}
    \label{tab:retrain}
\end{table}

\subsubsection{Knowledge Distillation.}
As outlined in the Experiments section of the main paper, we tested another alternative which uses poses optimized through a physics optimization scheme~\cite{yi2022physical, yi2025improving}, integrated with ground truth labels, as a form of label smoothing regularization (LSR)~\cite{yuan2020revisiting}. 
We show how this approach is suboptimal for our task as follows:

\begin{itemize}
    \item {\it Violation of Motion Properties}: Since the poses are optimized with a learning-based approach, it cannot guarantee adherence to the motion properties imposed by biomechanical properties (Properties 1-3).
    \item {\it Student Model Trade-off}: Although student models provide computational efficiency through lightweight inference, they inherently experience a reduction in accuracy compared to their teacher models. This trade-off compromises our objective of enhancing the performance and robustness of sparse-IMU-based motion capture. 
\end{itemize}

\end{document}